\documentclass[10pt,twocolumn,journal]{IEEEtran}
\usepackage{amsmath, graphics,amssymb,epsfig,subfigure,color,amsthm,cite}

\newtheorem{theorem}{Theorem}

\newtheorem{lemma}{Lemma}

\begin{document}

\sloppy
 
\title{ MAP Support Detection for Greedy Sparse Signal Recovery Algorithms in Compressive Sensing}
 
\author{
  \IEEEauthorblockN{Namyoon Lee,~ \IEEEmembership{Member,~IEEE}
% \\
%\normalsize Wireless Communication Research, Intel Labs\\ \normalsize Santa Clara, CA 95054, USA \\
%      { \normalsize E-mail~:~namyoon.lee@intel.com} 
\thanks{N. Lee is with Intel Labs, 2200 Mission College Blvd, Santa Clara, CA 95054, USA (e-mail:namyoon.lee@intel.com). This work was done when the author was with the University of Texas at Austin.}}
}

% 
%
%\author{
%  \IEEEauthorblockN{Namyoon Lee\\}
%\thanks{N. Lee is with Intel Labs, 
%2200 Mission College Blvd, Santa Clara, CA 95054, USA (e-mail:namyoon.lee@intel.com). This work was done when the author was with the University of Texas at Austin.}    
%}

\maketitle

\begin{abstract}
A reliable support detection is essential for a greedy algorithm to reconstruct a sparse signal accurately from compressed and noisy measurements. This paper proposes a novel support detection method for greedy algorithms, which is referred to as ``\textit{maximum a posteriori (MAP) support detection}". Unlike existing support detection methods that identify support indices with the largest correlation value in magnitude per iteration, the proposed method selects them with the largest likelihood ratios computed under the true and null support hypotheses by simultaneously exploiting the distributions of sensing matrix, sparse signal, and noise. Leveraging this technique, MAP-Matching Pursuit (MAP-MP) is first presented to show the advantages of exploiting the proposed support detection method, and a sufficient condition for perfect signal recovery is derived for the case when the sparse signal is binary. Subsequently, a set of iterative greedy algorithms, called MAP-generalized Orthogonal Matching Pursuit (MAP-gOMP), MAP-Compressive Sampling Matching Pursuit (MAP-CoSaMP), and MAP-Subspace Pursuit (MAP-SP) are presented to demonstrate the applicability of the proposed support detection method  to existing greedy algorithms. From empirical results, it is shown that the proposed greedy algorithms with highly reliable support detection can be better, faster, and easier to implement than basis pursuit via linear programming.

\end{abstract}

%%%%%%%%%%%%%%%%%%%%%%%%%%%%%%%%%
%%%%%%%%%%%%%%%%%%%%%%%%%%%%%%% 
\section{Introduction}

Compressive sensing (CS)  \cite{CandesTao2005,CandesRombergTao2006} is a technique to reconstruct sparse signals from compressed measurements. CS has received great attention due to its broad application areas including imaging, radar, and communication systems \cite{CandesWakin_Mag, Eldar_book}. The fundamental theory of CS guarantees to recover a high dimensional signal vector from linear measurements that are far fewer in number than the signal's dimension, provided that the sparsity of the signal, i.e. number of nonzero elements, is smaller than a certain fraction of the number of measurements. 

Denoting the sparse signal vector and the compressive sensing matrix as ${\bf x}\in\mathbb{R}^N$ and ${\bf \Phi}\in\mathbb{R}^{M\times N}$  respectively, with $M<N$, the optimal sparse recovery solution can be theoretically obtained by solving the $\ell_0$-minimization problem
\begin{align}
\min \|{\bf x}\|_0~~ {\rm subject}~{\rm to}~~ {\bf y}={\bf \Phi}{\bf x}. 
\end{align}
In practice, however, solving this problem is NP-hard \cite{Garey_book} and computationally unfeasible for large signal dimension ($N$).

Design of computationally efficient sparse signal recovery algorithms have extensively studied in past works. Basis Pursuit (BP) \cite{CandesRomberg2007,Candes2008,Chen_BP} is a representative sparse signal recovery algorithm leveraging convex optimization. Relaxing the $\ell_0$-minimization problem to a $\ell_1$-minimization problem, it has been shown that the sparse signal recovery problem can be solved with stability and uniform guarantees using linear programming, but with  polynomially bounded computation complexity. For example, an interior point method that solves the $\ell_1$-minimization problem has an associated  computational complexity of $\mathcal{O}(M^2N^3)$ \cite{Nesterov}.

As a result, greedy algorithms are also popular because their complexity is lower than that of BP although stability and  guarantees are challenging to prove \cite{TroppGilbert2007,Donoho,IHT2009,Needell2010_CoSaMP,Dai2009_SP,gOMP}. The underlying idea of greedy algorithms is to estimate the nonzero elements of a sparse vector iteratively. Orthogonal matching pursuit (OMP) is a well-known greedy algorithm\cite{TroppGilbert2007,DavenportWakin2010,CaiWang2011,Zhang,LiuTemlyakov2012}, which estimates the coordinate of the non-zero element in signal ${\bf x}$ that has the maximum absolute correlation between the column vector in the sensing matrix and the residual vector in each iteration. By subtracting the contribution from the measurement vector ${\bf y}$, the algorithm updates the entire support of ${\bf x}$ in an iterative manner. Although this algorithm is simple to implement, it is vulnerable to error propagation effect \cite{TroppGilbert2007,DavenportWakin2010,CaiWang2011,Zhang,LiuTemlyakov2012,Donoho,IHT2009}. This is because the OMP algorithm is not capable of removing incorrectly estimated supports once those are added to the support set during the iterations, which leads to significant performance degradation in the signal recovery.

Several other advanced greedy algorithms have been proposed to overcome the error propagation effect, which include Stagewise Orthogonal Matching Pursuit (StOMP) \cite{Donoho}, iterative hard thresholding (IHT) \cite{IHT2009}, generalized OMP (gOMP) \cite{gOMP},  Compressive Sampling Matching Pursuit (CoSaMP) \cite{Needell2010_CoSaMP}, and Subspace Pursuit (SP) \cite{Dai2009_SP}. The underlying principle of these advanced greedy algorithms is the selection of multiple support indices per iteration, leading to a decrease in the probability of estimating incorrect support elements. For example, in each iteration, StOMP \cite{Donoho} identifies  multiple support indices such that the correlation value in magnitude between the current residual vector and the corresponding column vector of ${\bf \Phi}$ exceeds a predefined threshold. Similarly, gOMP \cite{gOMP} chooses multiple supports that provide $L$ largest correlation in magnitude per iteration, where $L$ is a fixed parameter given in the algorithm. CoSaMP \cite{Needell2010_CoSaMP} and SP \cite{Dai2009_SP} also identify multiple support indices per iteration, but differ from StOMP and gOMP in that they perform a two-stage sparse signal estimation approach that allows to add or remove new support candidates adaptively. A common shortcoming of these greedy algorithms \cite{TroppGilbert2007,Donoho,DavenportWakin2010,Needell2010_CoSaMP,Dai2009_SP} is that they rely on the order statistics of the correlation value in magnitude for the support estimation.   

Depending on statistical distributions of sensing matrix, sparse signal, and noise, however, the selection of the index with the largest  correlation value may not be optimal in the sense of support detection probability. With this motivation, greedy algorithms called Bayesian matching pursuit were proposed in\cite{Ji2008,Schniter2009,Zayyani2009,Herzet2010,Herzet2011}. The key idea of Bayesian matching pursuit is the use of distributions of the sparse signal and noise in the support detection step. For example, fast Bayesian matching pursuit (FBMP) \cite{Schniter2009} performs sparse signal estimation via model selection, assuming a Gaussian distribution for the sparse vector. Similarly, in \cite{Herzet2010,Herzet2011} assuming the elements of a sparse signal are Bernoulli-Gaussian mixed variables, and a given deterministic sensing matrix, the algorithms jointly update a support index and the corresponding signal element at each iteration in order to maximize the increase of a local likelihood function. Although these approaches show a better sparse recovery performance compared to  conventional matching pursuit algorithms in the presence of noise, they are applicable to certain distributions of ${\bf x}$ like Bernoulli-Gaussian, and there are no provable performance guarantees.

In this paper, we continue the same spirit of harnessing the statistical distributions of sparse signal, sensing matrix, and noise for the support detection in greedy algorithms. Our main contribution is to propose a novel support detection method for greedy algorithms, which is referred to as maximum a posteriori (MAP) support detection. The key difference with prior work in \cite{Ji2008,Schniter2009,Zayyani2009,Herzet2010,Herzet2011} is that the proposed method estimates supports with the largest log-MAP ratio values computed under the true and null support hypotheses in each iteration by incorporating the distributions of the sensing matrix, the sparse signal, and noise jointly. Specifically, assuming the sensing matrix has elements that are drawn from independent and identically distributed (IID) Gaussian random variables, and the sparse signal has non-zero elements that follow an arbitrary distribution, the proposed method selects the support element having the maximum log-MAP ratio instead of selecting indices that exceed a certain threshold as in \cite{Schniter2009,Zayyani2009,Herzet2010,Herzet2011}. By leveraging this technique, we first present a novel greedy algorithm named ``\textit{MAP-Matching Pursuit  (MAP-MP)''} for the binary sparse signal reconstruction. Using this, it is shown that  MAP-MP exactly recovers a $K$-sparse binary signal within $K$ number of iterations \textit{almost surely}, provided that the number of measurement scales as 
\begin{align}
M=\mathcal{O}( (K+{\tilde \sigma}_w^2)\log(N)),
\end{align} 
where ${\tilde \sigma}_w^2$ is the normalized noise variance. This condition extends the existing statistical guarantees   proven in \cite{TroppGilbert2007} by incorporating a noise effect. Next, we extend our MAP-MP algorithm for the sparse signal with an arbitrary distribution using a moment matching technique. Subsequently, applying the proposed MAP support detection method, we propose a set of iterative greedy algorithms, called MAP-Orthogonal Matching Pursuit (MAP-OMP), MAP-generalized OMP (MAP-gOMP), MAP-Compressive Sampling Matching Pursuit (MAP-CoSaMP), and MAP-Subspace Pursuit (MAP-SP) to demonstrate the applicability of the proposed support detection method in improving the recovery performance of the existing algorithms. From the empirical results, it is shown that the proposed algorithms provide significant gains in the perfect recovery performance compared to that of the existing greedy algorithms as well as a  $\ell_1$-minimization algorithm via BP.

\vspace{-0.3cm} 
\section{Problem Statement}
 
We consider a sparse signal detection problem from compressed and noisy measurement. Let us denote a $N$ dimensional input signal vector by ${\bf x}\in \mathbb{R}^N$. We assume that the input vector is $K$-sparse, i.e., $\|{\bf x}\|_0=K \ll N$ and the sparsity level $K$ is known {\it a priori}. This prior information can be estimated accurately in some applications using the cross validation technique in\cite{Ward}. We denote the true support set by $\mathcal{T}\subset \{1,\ldots,N\}$ and $|\mathcal{T}|=K$. The non-zero entries of ${\bf x}$ are distributed according to a continuous distribution, i.e., $p({\bf x})=\prod_{k\in\mathcal{T}}p_{k}(x_k)$. Furthermore, we denote the sensing matrix consisting of $N$ column vectors by ${\bf \Phi}\in \mathbb{R}^{M \times N}$,
\begin{align}
{\bf \Phi} =\left[{\bf a}_1, {\bf a}_2, \ldots, {\bf a}_N\right],
\end{align}
where ${\bf a}_n$ denotes the $n$th dictionary vector whose entries are drawn from an IID Gaussian random distribution with zero mean and variance $\frac{1}{M}$, i.e., $\mathcal{N}\left(0,\frac{1}{M}\right)$. Then, the measurement equation is given by
\begin{align}
{\bf y}={\bf \Phi}{\bf x} +{\bf w}, \label{eq:system_eq}
\end{align}
where ${\bf y}\in \mathbb{R}^{M}$ and ${\bf w}\in \mathbb{R}^{M}$ are the measurement and noise vector, respectively. All entries of the noise vector are assumed to be IID Gaussian random variables with zero mean and variance $\sigma_w^2$, $\mathcal{N}\left(0,\sigma_w^2\right)$.

Throughout this paper, the difference between two sets  $\mathcal{T}$ and  $\mathcal{S}$ is denoted by $\mathcal{T}\setminus\mathcal{S}$. We use the subscript notations  ${\bf x}_{\mid \mathcal{S}}$ and ${\bf \Phi}_{\mid \mathcal{S}}$ to denote that vector ${\bf x}$ and matrix ${\bf \Phi}$ are being restricted to only elements or columns in set $\mathcal{S}$.  

%Possible applications using the non-zero elements drawn from the finite set of $\mathcal{C}$ are because many of wireless communication systems including multi-user detection for multiple access channels and device-to-device discovery can be modeled by the sparse linear system under such input constraint.

%\section{A Review: OMP}
%In this section, we briefly review the OMP algorithm.
%Orthogonal matching pursuit is a greedy algorithm that utilizes sequential forward selection to determine the signal representation within certain number of steps. At iteration k, OMP selects atom apk as part of signal representation such that
%
%
%OMP is a simple yet effective iterative greedy algorithm that widely used in compressive sensing. 

%\begin{table}
%\caption{OMP Algorithm}
%{\scalebox{1.20}{
%     \begin{tabular}{c|c}
%	\hline 
%	\hline
%	Input       & ${\bf A}$ and ${\bf y}$\\ \hline
%	Output       & $\mathcal{T}$ and ${\bf \hat x}$\\
%		\hline
%	Initialization       & $i:=0$, ${\bf r}^0:={\bf y}$, $\mathcal{T}^0:=\{\emptyset\}.$\\
%	     	\hline
%	While &  Repeat until a stopping criterion is met   \\
%	&  $i=:i+1$ \\
%	&  $Z_n^i=|{\bf a}_{n}^T{\bf r}^i|$ for $n\in[1:N]$ \\
%	&  $J^i=: \arg\max_{n} \left\{Z_{n}^i\right\}$\\
%	&  $\mathcal{T}^i=\mathcal{T}^{i-1}\cup J^i$ \\
%	&  ${\bf \hat x}^i:=\arg \min_{{\bf x}}\|{\bf A}_{\mathcal{T}^i}{\bf x}-{\bf y}\|_2$ \\
%	&  ${\bf r}^i:={\bf r}^{i-1}-{\bf A}_{\mathcal{T}^i}{\bf \hat x}^i$ \\
%	\hline \hline
%    \end{tabular}}}
%    \label{tab:OMP_alg}
%\end{table}

\section{MAP-Matching Pursuit}
In this section, we first present MAP-MP, a binary sparse signal ${\bf x}\in\{0,1\}^{N}$ recovery algorithm. Then, we derive a bound that provides a  sufficient condition for perfect signal recovery to demonstrate provable performance guarantees of the proposed algorithm.

\vspace{-0.3cm}
\subsection{Algorithm}
Similar to the other greedy algorithm \cite{TroppGilbert2007}, MAP-MP is a greedy algorithm that sequentially finds support indices and estimates the signal representation within a certain number of iterations. The core difference between the proposed MAP-MP algorithm and the prior OMP-type algorithms lies in the selection rule of the support index per iteration. In contrast to the OMP-type greedy algorithms, MAP-MP chooses the index based on a maximum likelihood hypothesis test by leveraging statistical property of the sensing matrix and the sparse signal. 

We begin by providing Lemmas that are required for explaining the MAP-MP algorithm. Lemma \ref{lemma1} provides the distribution of the inner product between two (atom) dictionary vectors generated by IID Gaussian random variable. Lemma \ref{lemma2} yields the distribution of the 2-norm of  each dictionary vector. Lemma \ref{lemma3}, in turn, provides an asymptotic behavior of the 2-norm of each dictionary vector when the measurement size $M$ goes to infinity.

\begin{lemma}\label{lemma1} Suppose that all the elements of ${\bf a}_n$ for $n\in[1:N]$ are drawn from IID Gaussian distribution with zero mean and variance $\frac{1}{M}$. Then, the distribution of $\frac{{\bf a}_n^T{\bf a}_{\ell}}{\|{\bf a}_n\|_2}$ is Gaussian with zero mean and variance $\frac{1}{M}$, i.e., $\frac{{\bf a}_n^T{\bf a}_{\ell}}{\|{\bf a}_n\|_2}\sim\mathcal{N}\left(0,\frac{1}{M}\right)$.
\end{lemma}
\begin{proof}
See Appendix \ref{proof:lemma1}.
\end{proof}

\begin{lemma}\label{lemma2} The distribution of the norm $\|{\bf a}_n\|_2$ is
\begin{align}
f_{\|{\bf a}_n\|_2}(x)=\frac{2^{1-\frac{M}{2}}M^{\frac{M}{2}}x^{M-1}e^{-\frac{Mx^2}{2}}}{\Gamma\left(\frac{M}{2}\right)}
\end{align}
and $\mathbb{E}[\|{\bf a}_n\|_2]=\sqrt{\frac{2}{M}}\frac{\Gamma\left(\frac{1+M}{2}\right)}{\Gamma\left(\frac{M}{2}\right)}$.
\end{lemma}
\begin{proof}
See Appendix \ref{proof:lemma2}.
\end{proof}

\begin{lemma}\label{lemma3} The norm $\|{\bf a}_n\|_2$ of each dictionary vector for $n\in[1:N]$ concentrates to one asymptotically as $M$ goes to infinity, 
\begin{align}
\lim_{M\rightarrow \infty}\mathbb{P}\left[ \left| \|{\bf a}_n\|_2 - 1 \right|  \geq \epsilon \right] =0
\end{align}
for some positive $\epsilon>0$. 
\end{lemma}
\begin{proof}
See Appendix \ref{proof:lemma3}.
\end{proof}

By leveraging these Lemmas, we explain the proposed algorithm. In the $k$th iteration, the algorithm produces $N$ correlation values $\left\{z^k_1,z_2^k,\ldots, z_N^k\right\}$ by computing the inner product between the residual vector ${\bf r}^{k-1}$ updated in the $(k\!-\!1)$th iteration and the $n$th column vector ${\bf a}_n$, i.e., $z^k_n= \frac{{\bf a}_n^T{\bf r}^{k-1}}{\|{\bf a}_n\|_2}$ for $n\in[1:N]$. Under the premise that the algorithm has perfectly found the elements of the true support, i.e., $\hat{x}_{\ell}=1$ for $\ell\in\mathcal{S}^{k-1}$, the residual vector is 
\begin{align}
{\bf r}^{k-1}=\sum_{\ell \in \mathcal{T}\setminus\mathcal{S}^{k-1} }{\bf a}_{\ell}x_{\ell}+{\bf w},
\end{align}
where ${\mathcal S}^{k-1}\subset \mathcal{T}$ and $|{\mathcal S}^{k-1}|=k-1$. Then, the inner product value $z^k_n= \frac{{\bf a}_n^T{\bf r}^{k-1}}{\|{\bf a}_n\|_2}$ can be expressed as a linear combination of the remaining non-zero elements and their corresponding supports as follows:
\begin{align}
z_n^k%&=\frac{{\bf a}_n^T{\bf r}^{k-1}}{\|{\bf a}_n\|_2} \nonumber\\
&=\frac{{\bf a}_n^T}{\|{\bf a}_n\|_2}\left(\sum_{\ell \in \mathcal{T}\setminus\mathcal{S}^{k-1}}\!\!\!\!{\bf a}_{\ell}{ x}_{\ell}+{\bf w}\right) \nonumber\\
&=\|{\bf a}_{n}\|_2{ x}_{n}+\!\!\!\!\!\sum_{\ell \in \mathcal{T}\setminus\{\mathcal{S}^{k-1}\cup\! \{n\}\}}\frac{{\bf a}_n^T{\bf a}_{\ell}{ x}_{\ell}}{\|{\bf a}_n\|_2} +\frac{{\bf a}_n^T{\bf w}}{\|{\bf a}_n\|_2}.\label{eq:corr}%&={ x}_{n}+\!\!\!\!\!\sum_{\ell \in \mathcal{T}\setminus\{\mathcal{S}^{k-1}\cup \{n\}\}}{\bf a}_n^T{\bf a}_{\ell}{ x}_{\ell}  +{\bf a}_n^T{\bf w}
\end{align}
Using (\ref{eq:corr}), the proposed MAP-MP algorithm performs the hypothesis test with two hypotheses corresponding to $x_n=0$ and $x_n= 1$, respectively, as follows:
\begin{align}
&\mathcal{H}_0: z_n^k=\sum_{\ell \in \mathcal{T}\setminus\{\mathcal{S}^{k-1}\}}\!\!\frac{{\bf a}_n^T{\bf a}_{\ell}}{\|{\bf a}_n\|_2}x_{\ell}   +\frac{{\bf a}_n^T{\bf w}}{\|{\bf a}_n\|_2}\label{eq:HT_binary_0}\\
&\mathcal{H}_{1}:z_n^k= \|{\bf a}_{n}\|_2x_n+\!\!\!\!\!\!\!\!\!\sum_{\ell \in \mathcal{T}\setminus\{\mathcal{S}^{k-1}\cup \{n\}\}} \frac{{\bf a}_n^T{\bf a}_{\ell}}{\|{\bf a}_n\|_2}x_{\ell}+\frac{{\bf a}_n^T{\bf w}}{\|{\bf a}_n\|_2}, \label{eq:HT_binary_1}
\end{align}
where ￼$\mathcal{H}_0$ is the null hypothesis such that the $n$th column vector ${\bf a}_n$ is not the support, i.e., $x_n=0$ ($n\notin \mathcal{T}$)￼, and $\mathcal{H}_{1}$ is the alternate hypothesis indicating that the $n$th column vector is a non-zero support and the corresponding signal value is $1$, i.e., $x_n=1$ ($n\in \mathcal{T}$)￼. These two  hypotheses in \eqref{eq:HT_binary_0} and \eqref{eq:HT_binary_1} involve multiple levels of randomness, namely, 
\begin{enumerate}
\item The randomness associated with the inner product between two distinct vectors $\frac{{\bf a}_{n}}{\|{\bf a}_n\|_2}$ (unit norm) and ${\bf a}_{\ell}$; this is distributed as a Gaussian random variable, i.e., $\frac{{\bf a}_n^T{\bf a}_{\ell}}{\|{\bf a}_n\|_2}\sim\mathcal{N}\left(0,\frac{1}{M}\right)$ for $\ell\neq n$ as shown in Lemma \ref{lemma1} (See Appendix). 
\item The randomness associated with the effective noise $\frac{{\bf a}_n^T{\bf w}}{\|{\bf a}_n\|_2}$; this is Gaussian with zero mean and variance $\sigma_w^2$ i.e., $\frac{{\bf a}_n^T{\bf w}}{\|{\bf a}_n\|_2}\sim\mathcal{N}\left(0,\sigma_w^2\right)$, as ${\bf w}$ is isotropically distributed in ${\mathbb R}^M$.
\item  The randomness associated with the sum of independent Gaussian random variables, $z_n^k=\sum_{\ell \in \mathcal{T}\setminus\{\mathcal{S}^{k-1}\cup \{n\}\}}\frac{{\bf a}_n^T{\bf a}_{\ell}}{\|{\bf a}_n\|_2}+\frac{{\bf a}_n^T{\bf w}}{\|{\bf a}_n\|_2}$; this is also Gaussian with zero mean and variance $\mathbb{E}\left[\left(z_n^k\right)^2\right]=\frac{K-(k-1)}{M}+\sigma_w^2$ as  $\frac{{\bf a}_n^T{\bf a}_{\ell}}{\|{\bf a}_n\|_2}$, $\frac{{\bf a}_n^T{\bf a}_{j}}{\|{\bf a}_n\|_2}$, and $\frac{{\bf a}_n^T{\bf w}}{\|{\bf a}_n\|_2}$ are mutually independent Gaussian random variables for $\ell \neq j$.
\item The randomness associated with the norm of the column vector $\|{\bf a}_n\|_2$; this is a scaled Chi-distribution with $M$ degrees of freedom, i.e., $f_{\|{\bf a}_n\|_2}(x)=\frac{2^{1-\frac{M}{2}}M^{\frac{M}{2}}x^{M-1}e^{-\frac{Mx^2}{2}}}{\Gamma\left(\frac{M}{2}\right)}$ as shown in Lemma \ref{lemma2}.
\end{enumerate}
 
Using these facts, the conditional distribution of $z_n^k$ under the null hypothesis is given by
\begin{align}
\mathbb{P}\left(z_{n}^k|x_{n}=0\right)&=\frac{1}{\sigma_{0}\sqrt{2\pi}}\exp\left(-\frac{|z_{n}^k|^2}{2\sigma_{0}^2}\right), \label{eq:PDF_null}
\end{align}
where $\sigma_{0}=\sqrt{\frac{K-(k-1)}{M}+\sigma_w^2}$. Similarly, under the hypothesis of $x_n=1$ and $\|{\bf a}_n\|_2=u$, the conditional distribution of $z^k_n$ is Gaussian with mean $u$ and variance $\frac{K-(k-1)+1}{M}+\sigma_w^2$, i.e.,
\begin{align}
\mathbb{P}\left(z_{n}^k|x_{n}\!=\!1, \|{\bf a}_n\|_2\!=\!u\right)&=\frac{\exp\left(-\frac{|z_{n}^k-u|^2}{2\sigma_{1}^2}\right)}{\sigma_{1}\sqrt{2\pi}},\label{eq:PDF_notnull}
\end{align}
where $\sigma_{1}=\sqrt{\frac{K-(k-1)+1}{M}+\sigma_w^2}$. From Lemma \ref{lemma2}, by marginalizing the conditional distribution in (\ref{eq:PDF_notnull}) with respect to $u$, we obtain the conditional distribution under the hypothesis of $x_n=1$ as
\begin{align}
\mathbb{P}\left(z_{n}^k|x_{n}=1\right)&=\mathbb{E}_{\|{\bf a}_n\|_2}\left[ \mathbb{P}\left(z_{n}^k|x_{n}=1, \|{\bf a}_n\|\right) \right] \nonumber \\
&=\int_{0}^{\infty}\frac{e^{-\frac{|z_{n}^k-u|^2}{2\sigma_{1}^2}}}{\sigma_{1}\sqrt{2\pi}}\frac{2^{1-\frac{M}{2}}M^{\frac{M}{2}}u^{M-1}e^{-\frac{Mu^2}{2}}}{\Gamma\left(\frac{M}{2}\right)} {\rm d} u.\label{eq:PDF_notnull_mar}
\end{align}
This conditional distribution is intractable to analyze due to the integral expression. Applying Jensen's inequality, we obtain a lower bound of the conditional distribution function in a closed-form as follows:
\begin{align}
\mathbb{P}\left(z_{n}^k|x_{n}\!=\!1\right)&\geq \frac{\exp\left(- \frac{\mathbb{E}\left[(z_{n}^k-\|{\bf a}_n\|_2)^2\right]}{2\sigma_{1}^2}\right) }{\sigma_{1}\sqrt{2\pi}}\nonumber \\
&\geq\frac{\exp\left(- \frac{(z_{n}^k-\mathbb{E}\left[\|{\bf a}_n\|_2\right])^2}{2\sigma_{1}^2}\right)}{\sigma_{1}\sqrt{2\pi}} \nonumber \\
&=\frac{\exp\left(-\frac{\left(z_{n}^k-\sqrt{\frac{2}{M}}\frac{\Gamma\left(\frac{1+M}{2}\right)}{\Gamma\left(\frac{M}{2}\right)}\right)^2}{2\sigma_{1}^2}\right)}{\sigma_{1}\sqrt{2\pi}}. \label{eq:PDF_notnull_mar_lowerbound}
\end{align}
where the first and the second inequalities follow from the facts that $e^{-x}$ and $(a-x)^2$ are convex functions with respect to $x$ for any $a$, respectively. The last equality is because $\mathbb{E}[\|{\bf a}_n\|_2]=\sqrt{\frac{2}{M}}\frac{\Gamma\left(\frac{1+M}{2}\right)}{\Gamma\left(\frac{M}{2}\right)}$ as shown in Lemma \ref{lemma2}. From Lemma \ref{lemma3}, it is shown that this lower bound becomes tight, as the distribution of $\|{\bf a}_n\|_2$ converges to its mean value $\lim_{M\rightarrow \infty}\sqrt{\frac{2}{M}}\frac{\Gamma\left(\frac{1+M}{2}\right)}{\Gamma\left(\frac{M}{2}\right)}=1$ almost surely. As a result, for large enough $M$, the conditional distribution under the hypothesis of $x_n=1$ is simply approximated as
\begin{align}
\mathbb{P}\left(z_{n}^k|x_{n}=1\right)&\simeq \frac{1}{\sigma_{1}\sqrt{2\pi}}\exp\left(-\frac{|z_{n}^k-1|^2}{2\sigma_{1}^2}\right). \label{eq:PDF_notnull_mar_sim}
\end{align}
Leveraging the conditional probability density functions in (\ref{eq:PDF_null}) and (\ref{eq:PDF_notnull_mar_sim}), the MAP ratio for a given observation $z_n^k$ is 
\begin{align}
\Lambda \left(z_{n}^k\right)&=\ln\left(\frac{ \mathbb{P}\left(n\in\mathcal{T}\mid  z_{n}^k\right)}{\mathbb{P}\left(n\notin\mathcal{T}\mid  z_{n}^k\right)}\right)\nonumber \\
&\stackrel{(a)}{=}\ln\left(\frac{ \mathbb{P}\left(z_{n}^k|n\in\mathcal{T}\right)\mathbb{P}\left(n\in\mathcal{T}\right)}{\mathbb{P}\left( z_{n}^k|n\notin\mathcal{T}\right)\mathbb{P}\left(n\notin\mathcal{T}\right)}\right)\nonumber \\
&=\ln\left(\frac{ \frac{1}{\sigma_{1}\sqrt{2\pi}}\exp\left(-\frac{|z_{n}^k-1|^2}{2\sigma_{1}^2}\right) }{\frac{1}{\sigma_{0}\sqrt{2\pi}}\exp\left(-\frac{|z_{n}^k|^2}{2\sigma_{0}^2}\right) }\right)+\ln\left(\frac{  \mathbb{P}\left(n\in\mathcal{T}\right)}{ \mathbb{P}\left(n\notin\mathcal{T}\right)}\right)\nonumber \\
%&=\ln\left(\frac{ \mathbb{P}\left( z_{n}^k|n\in\mathcal{T}\right) }{\mathbb{P}\left( z_{n}^k|n\notin\mathcal{T}\right) }\right)+\ln\left(\frac{  \mathbb{P}\left(n\in\mathcal{T}\right)}{ \mathbb{P}\left(n\notin\mathcal{T}\right)}\right)\nonumber \\
&\stackrel{(b)}{=}\frac{(z_n^k)^2}{2\sigma_0^2}-\frac{(z_n^k-1)^2}{2\sigma_1^2}+\ln\left(\frac{\sigma_0}{\sigma_1}\right)+\ln\left(\frac{K}{N-K}\right) \nonumber \\
&=\frac{(z_n^k)^2}{2\frac{K-k+1}{M}+2\sigma_w^2 }-\frac{(z_n^k-1)^2}{2\frac{K-k}{M}+2\sigma_w^2 } \nonumber \\
&+\frac{1}{2}\ln\left(\frac{(K\!-\!k\!+\!1)\!+\!M\sigma^2_w}{(K\!-\!k)\!+\!M\sigma^2_w}\right)\!+\!\ln\left(\frac{K}{N-K}\right),\label{eq:MAP_Ratio}
\end{align}
where (a) follows from the Bayes' rule and (b) comes from the assumption that the $K$ non-zero supports are uniformly distributed from $1$ to $N$. This log likelihood ratio value carries reliability information about how the $n$th column vector in the sensing matrix is likely to belong to the true support in the $k$th iteration. Accordingly, at iteration $k\in\{1,\ldots,K-1\}$, the proposed MAP-MP algorithm selects index $J^k$ that maximizes $\Lambda \left(z_{n}^k\right)$, namely,
\begin{align}
J^k&=\arg\max_{n\in[1:N]}\Lambda(z_{n}^k) \nonumber \\
&=\arg\max_{n\in[1:N]} \frac{(z_n^k)^2}{\frac{K-k+1}{M}+\sigma_w^2 }-\frac{(z_n^k-1)^2}{\frac{K-k}{M}+\sigma_w^2 }. 
\end{align} 
Once index $J^k$ is selected, MAP-MP estimates the new sparse representation ${\bf \hat x}^k$ using the updated support set $\mathcal{S}^k=\mathcal{S}^{k-1}\cup \{J^k\}$. Since the signal is assumed to be a binary, the new sparse representation is set to be one, namely,
\begin{align}
{\bf \hat x}^k_{\mathcal{S}^k}=1.
\end{align}
Lastly, to remove the contribution of ${\bf \hat x}^k_{\mathcal{S}^k}$, we update the new residual signal such that
\begin{align}
{\bf r}^k={\bf y}-{\bf \Phi}_{\mid \mathcal{S}^k}{\bf \hat x}^k_{\mathcal{S}^k}.
\end{align} 
 
%
%\begin{table}\label{table1}
%\caption{GBD Algorithm}\center
%{\scalebox{1.0}{\begin{tabular}{|l|}\hline\hline
%1)  Initialization:\\
%\hspace{5mm} $k:=0$, $\mathbf{\hat x}^{0} = {\bf 0}$ \\
%\hspace{5mm} ${\bf r}^0:={\bf y}$ (the current residual)\\
%\hspace{5mm}  $\mathcal{S}^0:=\{\emptyset\}.$\\
%\hline
%2) Repeat until a stopping criterion is met\\
%\hspace{5mm} i) $k:=k+1$. \\
%\hspace{5mm} ii) Compute the current proxy: \\
%\hspace{10mm} $z^k_n = {\bf a}^T_n\mathbf{r}^{k-1}/\| {\bf a}_n\|_2$ for $n\in[1:N]$.\\
%\hspace{5mm} iii) Select the largest index of MAP ratio:\\
%\hspace{10mm} $J^k=: \arg\max_{n} \left\{\mathcal{L}(z_{n}^k)\right\}$.\\
%\hspace{5mm} iv) Merge the support set:\\
%\hspace{10mm} $\mathcal{S}^k=\mathcal{S}^{k-1}\cup J^k$.\\
%\hspace{5mm} v) Update sparse signal (Binary signal case):\\ 
%\hspace{10mm} ${\bf \hat x}^k_{\mathcal{S}^k}:=1$.\\
%\hspace{5mm} vi) Update the residual for next round:\\
%\hspace{10mm} ${\bf r}^k:={\bf y}-{\bf A}_{\mid \mathcal{S}^k}{\bf \hat x}^k_{\mathcal{S}^k}$.\\
%\hline\hline
%\end{tabular}}}
%\end{table}

%
\subsection{Remarks}

To obtain more insight on the proposed support detection method, it is instructive to consider certain special cases.

{\bf Noise-Free Case}: Let us consider the case of noise-free compressive sensing, i.e., $\sigma_w^2=0$. The log-MAP ratio boils down to
  \begin{align}
\Lambda(z_{n}^k)&=  \frac{M(z_n^k)^2} {2(K-k+1) }-\frac{M(z_n^k-1)^2}{ 2(K-k) } \nonumber \\&+\frac{1}{2}\ln\left(\frac{K-k+1 }{K-k }\right)+\ln\left(\frac{K}{N-K}\right). \label{eq:LLR_binary_noisefree}
\end{align}
This expression clearly shows that the MAP ratio in the $k$th iteration is a function of the relevant system parameters-the dimension of the measurement vector $M$ and the sparsity level $K$. One key property of the proposed algorithm is that it updates the log-MAP ratio adaptively, since the variances of the conditional probability density functions decrease under the premise that the algorithm successively estimates the signal at each iteration. For the noise-free case, in the last iteration $k=K$, we slightly need to modify the computation of the ratio, as $\mathbb{P}\left( z_{n}^K|n\in\mathcal{T}\right)=1$. Accordingly, the modified ratio in the last iteration for the noise-free case is given by
\begin{align}
\Lambda(z_{n}^K)&=  \frac{M(z_n^K)^2} {2  } +\ln\left(\frac{K}{N-K}\right).
\end{align} 

{\bf High Noise Power Case}: Let us consider the high noise power scenario, i.e., $\sigma_w^2 \gg \frac{K}{M}$. In this case, the MAP ratio in (\ref{eq:MAP_Ratio}) is approximated as 
\begin{align}
\Lambda(z_{n}^k)
&\simeq\frac{(z_n^k)^2}{2\sigma_w^2 }-\frac{(z_n^k-1)^2}{ 2\sigma_w^2 }=\frac{2z_n^k-1}{2\sigma_w^2}.\label{eq:ML_binary_highnoise}
\end{align}
From this, we are able to observe that the selection of the largest index of the MAP ratio is equivalent to the selection of the largest index of the correlation value $z_{n}^k$ in the high noise power regime, namely,
\begin{align}
\arg \max_n\Lambda(z_{n}^k)=\arg \max_n z_n^k.\end{align}  
Therefore, the conventional support detection methods that select the largest correlation value $z_n^k$ is the optimal in the sense of the MAP detection strategy for the high noise power regime. For the cases of low noise power and noise-free, however, the selection of the largest absolute value of $z_n^k$ for the support detection is not optimal. This fact clearly exhibits the benefits of the proposed MAP-MP against the conventional OMP algorithm in \cite{TroppGilbert2007}.

\subsection{Asymptotic Analysis for Exact Recovery}
In this section, we derive a lower bound of the required measurements for the exact support recovery when the proposed MAP-MP is applied for the binary sparse signal. Unlike the prior analysis approaches that rely on the Restricted Isometry Property (RIP) \cite{Candes2008,DavenportWakin2010,CaiWang2011} or an information theoretical analysis tool in \cite{Wainwright2009}, we directly compute a lower bound of the success probability that the proposed algorithm identifies the $K$-sparse binary signal within $K$ number of iterations. Utilizing this, a lower bound of the required measurements is derived to reconstruct the signal perfectly as the signal dimension approaches infinity. The following theorem shows the main analysis result.

\begin{theorem} \label{Theorem1} The proposed MAP-MP algorithm perfectly recovers a $K$-sparse binary sparse vector, ${\bf x}\in\{0,1\}^{N}$, with $M$ noisy measurements within $K$ number of iterations, provided that the number of measurements scales as
 \begin{align}
M=\mathcal{O}\left((K+{\tilde \sigma}_w^2)\ln(N)\right),
\end{align}
when $N$ and $K$ go to infinity. Here, ${\tilde \sigma}_w^2$ denotes a normalized noise variance defined as ${\tilde \sigma}_w^2=\frac{\sigma_w^2}{M}$.
\end{theorem}

\begin{proof}
Without loss of generality, we assume that the first $K$ columns are the true supports, i.e., $x_n =1$ for $n\in[1:K]$, i.e., $\mathcal{T}=\{1,2,\ldots,K\}$ and the remaining $N-K$ columns are the zero supports. Furthermore, we denote $E_s^k$ to be the success recovery probability event in the $k$th iteration. Then, the success recovery probability of the $K$-sparse signal within $K$ number of iterations is given by 
\begin{align}
{ P}_s&= \mathbb{P}\left(\cap_{k=1}^KE_s^k\right) \nonumber \\
&=\mathbb{P}(E_s^1)\mathbb{P}(E_s^2|E_s^1)\times \cdots \times \mathbb{P}(E_s^K|E_s^{K-1},\ldots,E_s^1),
\end{align}
where the equality comes from the probability chain rule.
To prove that $P_s$ approaches one asymptotically as $N\rightarrow \infty$, it suffices to check that the algorithm correctly identifies the column of the true support in the $k$th iteration conditioned that all the prior iterations recover the true supports successfully, i.e., $\mathbb{P}(E_s^k|E_s^{k-1},\ldots,E_s^1)=1-o\left(\frac{1}{K}\right)$ as $N \rightarrow \infty$ for any $k\in[1:K]$.

To detect the support correctly in the $k$th iteration of the proposed algorithm, the maximum of $\Lambda(z^k_{\ell})$ for $\ell\in \mathcal{T}\setminus\mathcal{S}^k$ should be larger than the maximum of $\Lambda(z^k_{n})$ for $n \in \mathcal{T}^c=\{K+1,\ldots,N\}$, which is 
\begin{align}
\mathbb{P}(E_s^k|E_s^{k\!-\!1},\ldots,E_s^1)=\! \mathbb{P}\left[\!\max_{\ell\in \mathcal{T}\setminus\mathcal{S}^k}\!\!\Lambda\left(z_{\ell}^k\right) \!\geq \!\max_{n \in\mathcal{T}^c } \Lambda\left(z_{n}^k\right) \!\right].  \nonumber
\end{align}
By selecting an arbitrary element of $\ell \in  \mathcal{T}\setminus\mathcal{S}^k$, a lower bound of the success probability in the $k$th iteration is given by
\begin{align}
\mathbb{P}(E_s^k|E_s^{k-1},\ldots,E_s^1)\!&\geq  \mathbb{P}\left[ \Lambda\left(z_{\ell}^k\right) \geq \max_{n \in \mathcal{T}^c} \Lambda\left(z_{n}^k\right) \right] \nonumber \\
&=\prod_{n=1}^{N-K}\mathbb{P}\left[\Lambda(z_{\ell}^k) \geq \Lambda(z_{n}^k) \right] \nonumber \\
%&= \prod_{n=1}^{N-K}\left(1-\mathbb{P}\left[ \Lambda(z_{n}^k) < \Lambda(z_{\ell}^k) \right]\right) \nonumber \\
&=\!\left(\!1\!-\!\mathbb{P}\left[ \Lambda(z_{\ell}^k) < \Lambda(z_{N}^k) \right]\right)^{N-K}, \label{eq:lowbound}
\end{align}
where the first equality follows from the fact that $\{\Lambda(z_{K+1}^k), \ldots, \Lambda(z_{N}^k)\}$ are mutually independent as $\{z_{K+1}^k,\ldots,z_N^k\}$ are IID Gaussian random variables with zero mean and variance $\sigma_0^2$. To this end, we need to compute the probability that $\Lambda(z_{\ell}^k)$ is less than $ \Lambda(z_{N}^k)$ as follows:
\begin{align}
&\mathbb{P}\left[\Lambda(z_{N}^k) >\Lambda(z_{\ell}^k) \right] \nonumber \\&= \mathbb{P}\left[\frac{(z_{N}^k)^2}{2\sigma^2_0}-\frac{(z_{N}^k-1)^2}{2\sigma^2_1}>\frac{(z_{\ell}^k)^2}{2\sigma^2_0}-\frac{(z_{\ell}^k-1)^2}{2\sigma^2_1}\right] \nonumber\\
&=\mathbb{P}\left[e^{-\lambda\left(\frac{\left(z_{N}^k\right)^2}{2\sigma^2_0}-\frac{\left(z_{N}^k-1\right)^2}{2\sigma^2_1}\right)}< e^{-\lambda\left(\frac{\left(z_{\ell}^k\right)^2}{2\sigma^2_0}-\frac{\left(z_{\ell}^k-1\right)^2}{2\sigma^2_1}\right)}\right] \nonumber \\
&\leq \min_{\lambda\geq 0} \mathbb{E} \left[e^{\lambda\left(\frac{\left(z_{N}^k\right)^2}{2\sigma^2_0}-\frac{\left(z_{N}^k-1\right)^2}{2\sigma^2_1}\right)}\right] \mathbb{E} \left[e^{-\lambda\left(\frac{\left(z_{\ell}^k\right)^2}{2\sigma^2_0}-\frac{\left(z_{\ell}^k-1\right)^2}{2\sigma^2_1}\right)}\right] \label{eq:uppr_error}
\end{align}
where the last inequality follows from Markov's inequality and the independence of $z_{\ell}^k$ and $z_{N}^k$. Since $z_{N}^k$ given $x_{N}=0$ is distributed as in (\ref{eq:PDF_null}), the first term in (\ref{eq:uppr_error}) is calculated as 
\begin{align}
 \mathbb{E}\left[\!e^{\lambda\left(\!\frac{\left(z_{N}^k\right)^2}{2\sigma^2_0}-\frac{\left(z_{N}^k-1\right)^2}{2\sigma^2_1}\!\right)}\!\!\right] &= \int_{-\infty}^{\infty}\frac{e^{-\frac{t^2}{2\sigma_0^2}}e^{\lambda\left(\frac{t^2}{2\sigma^2_0}-\frac{(t-1)^2}{2\sigma^2_1}\right)}}{\sqrt{2\pi}\sigma_0} {\rm d}t
\nonumber \\
&= \frac{e^{\frac{-\lambda(1-\lambda)}{2\lambda\left(\sigma_0^2-\sigma_1^2\right)+2\sigma_1^2}}}{\sigma_0\sqrt{\frac{1-\lambda}{\sigma_0^2}+\frac{\lambda}{\sigma_1^2}}}.
\end{align}
Similarly, using the distribution of $z_{\ell}^k$ given $x_{\ell}=1$ in (\ref{eq:PDF_notnull}), the second term in (\ref{eq:uppr_error}) is computed as
\begin{align}
\mathbb{E}\left[\!e^{-\lambda\left(\!\frac{\left(z_{\ell}^k\right)^2}{2\sigma^2_0}-\frac{\left(z_{\ell}^k-1\right)^2}{2\sigma^2_1}\!\right)}\!\!\right]  &=\!\! \int_{-\infty}^{\infty}\frac{e^{-\frac{(t-1)^2}{2\sigma_1^2}}e^{-\lambda\left(\frac{t^2}{2\sigma^2_0}-\frac{(t-1)^2}{2\sigma^2_1}\right)}}{\sqrt{2\pi}\sigma_1} {\rm d}t \nonumber \\
&= \frac{e^{\frac{\lambda(1-\lambda)}{2\lambda\left(\sigma_0^2-\sigma_1^2\right)-2\sigma_0^2}}}{\sigma_1\sqrt{\frac{\lambda}{\sigma_0^2}+\frac{1-\lambda}{\sigma_1^2}}}.
\end{align}
Plugging $\lambda=\frac{1}{2}>0$, the probability that the MAP ratio under the zero support is greater than that under the non-zero support is upper bounded by
\begin{align}
\mathbb{P}\left[\Lambda(z_{N}^k) >\Lambda(z_{\ell}^k) \right]&\leq \frac{e^{\frac{-1}{2(\sigma_0^2+\sigma_1^2)}}}{\frac{1}{2}(\frac{\sigma_1}{\sigma_0}+\frac{\sigma_0}{\sigma_1})}.
\label{eq:uppr_error2}
\end{align}
Since $\sigma_0^2=\frac{K-k+1+{\tilde \sigma}_w^2}{M}$ and $\sigma_1^2=\frac{K-k+{\tilde \sigma}_w^2}{M}$ in the $k$th iteration, this error upper bound is further simplified as 
\begin{align}
\mathbb{P}\left[\Lambda(z_{N}^k) >\Lambda(z_{\ell}^k) \right]&\leq \frac{e^{\frac{-M}{2(2K-2k+2{\tilde \sigma}_w^2+1)}}}{\frac{1}{2}(\frac{K-k+{\tilde \sigma}_w^2}{K-k+1+{\tilde \sigma}_w^2}+\frac{K-k+1+{\tilde \sigma}_w^2}{K-k+{\tilde \sigma}_w^2})} \nonumber \\
&\leq e^{\frac{-M}{2(2K-2k+2{\tilde \sigma}_w^2+1)}}.
\label{eq:uppr_error2}
\end{align}
Plugging (\ref{eq:uppr_error2}) into (\ref{eq:lowbound}), we have a lower bound as follows:
\begin{align}
\mathbb{P}(E_s^k|E_s^{k-1},\ldots,E_s^1)\geq \left(1-e^{\frac{-M}{2(2K-2k+2{\tilde \sigma}_w^2+1)}}\right)^{N-K}. \label{eq:31}
\end{align}
From \eqref{eq:31}, we observe that the success probability in the first iteration is lower than that of any other remaining iterations, i.e., $\mathbb{P}(E_s^1)\leq \mathbb{P}(E_s^k|E_s^{k-1},\ldots,E_s^1)$ for $\forall k$. It follows that the lower bound of the exact recovery probability is
\begin{align}
{ P}_s&=\mathbb{P}(E_s^1)\mathbb{P}(E_s^2|E_s^1)\times \cdots \times \mathbb{P}(E_s^K|E_s^{K-1},\ldots,E_s^1) \nonumber \\
&\geq \left(1-e^{\frac{-M}{2(2K-1+2{\tilde \sigma}_w^2)}}\right)^{K(N-K)}. \label{eq:32}
\end{align}
Assuming $M= (1+\epsilon)2(2K-1+2{\tilde \sigma}_w^2)\ln(K(N-K))$, the lower bound is rewritten as
\begin{align}
\ln\left(P_s\right)\geq K(N\!-\!K)\!\ln\!\left(\!1\!-\!\frac{1}{\left\{K(N\!-\!K)\right\}^{1+\epsilon}}\!\right). \nonumber
\end{align}
Let $K=\delta N$ for a positive $\delta >0$. Then, as $N$ goes to infinity, we have
\begin{align}
\lim_{N \rightarrow \infty}\ln\left(P_s\right)&\geq\lim_{N \rightarrow \infty} N^2\delta(1\!-\!\delta)\!\ln\!\left(\!1\!-\!\frac{1}{\left\{N^2\delta(1\!-\!\delta)\right\}^{1+\epsilon}}\!\right) \nonumber \\
&=\lim_{N \rightarrow \infty}\frac{4(1+\epsilon)\delta(1\!-\!\delta)}{\left\{N^2\delta(1\!-\!\delta)\right\}^{1+\epsilon}-1}=0,
\end{align}
where the second equality follows from L'Hospital's rule. Consequently, we conclude that $\lim_{N \rightarrow \infty}P_s=1$. From the facts that $N>M>2K$ (the condition for a unique sparse solution) and $\ln(K(N-K))=\ln(N-K) +\ln(K) \leq 2\ln(N-K)$, it is possible that the $K$-sparse binary signal is perfectly recovered within $K$ number of iterations, if the number of measurements scales as, at least, $M\geq (1+\epsilon)2(4K-2+4{\tilde \sigma}_w^2)\ln(N-K)$ for some $\epsilon >0$. Therefore, the scaling law of the required number of measurements becomes $M=\mathcal{O}\left((K+{\tilde \sigma}_w^2)\ln(N)\right)$, which completes the proof.
\end{proof}

Theorem 1 shows the statistical guarantee of the proposed MAP-MP algorithm for the binary signal. The guarantee is that the proposed MAP-MP algorithm recovers the $K$-sparse binary signal perfectly with $K$ number of iterations, if the number of (noisy) measurements scales as $\mathcal{O}\left((K+{\tilde \sigma}_w^2)\ln(N)\right)$. This measurement scaling law clearly exhibits that the required measurements should increase with the sparsity level $K$ and the normalized noise variance ${\tilde \sigma}_w^2$. This result backs the intuition that the measurements should increase the sparsity level and noise variance linearly. Meanwhile, the requirement measurements increase with $N$ logarithmically. This condition extends the existing statistical guarantee for OMP proven in \cite{TroppGilbert2007} by incorporating noise effects. %Theorem 1 suggests that the GBD algorithm guarantees the perfect $K$ sparse signal recovery almost surely within $K$ number of iterations, which also allows to bound the computational complexity for the proposed GBD algorithm.

\section{MAP-OMP}
In the previous section, we have proposed a new sparse signal recovery algorithm, assuming a binary sparse signal. In some applications, however, the component of the non-zero support can be an arbitrary value drawn from a continuous probability distribution $f_{x_n}(u)$. In this section, we present a modified MAP-MP algorithm for the sparse signal whose non-zero element is distributed according to a distribution $f_{x_n}(u)$, which is referred to as MAP-OMP. In contrast to the MAP-MP algorithm, MAP-OMP uses an orthogonal projection method when the unknown elements are estimated, which causes estimation errors. Therefore, it is essential for characterizing statistical properties of the estimation errors in each iteration in order to apply a hypotheses test. The following lemma shows the statistical properties of the estimation errors. 

\begin{lemma}\label{lemma4} Let ${\hat {\bf x}}_{|\mathcal{S}^{k}}=\left({\bf \Phi}_{|\mathcal{S}^{k}}^T{\bf \Phi}_{|\mathcal{S}^{k}}\right)^{-1}{\bf \Phi}_{|\mathcal{S}^{k}}^T{\bf y}$ be the estimate of ${{\bf x}}_{|\mathcal{S}^{k}}$. Then, the mean vector and the covariance matrix of estimation error, ${\bf e}_x={ {\bf \hat x}}_{|\mathcal{S}^{k}}-{{\bf x}}_{|\mathcal{S}^{k}}$, are 
\begin{align}
\mathbb{E}[{\bf e}_x]={\bf 0}~ {\rm and} ~~\mathbb{E}\!\left[{\bf e}_x{\bf e}_x^T\right]=\frac{\sigma^2_x(K-k)+\tilde\sigma^2_w}{M-k-2}{\bf I}.
\end{align} 
\end{lemma}
\begin{proof}
See Appendix \ref{proof:lemma4}.
\end{proof}\vspace{-0.1cm}

Using this lemma, we explain the MAP-OMP algorithm. Let ${\hat x}_{i}$ be the estimate of $x_{i}$ in the $(k\!-\!1)$th iteration where $i\in\mathcal{S}^{k-1}$. Then, the residual vector is 
\begin{align}
{\bf r}^{k-1}&={\bf y}-{\bf \Phi}_{\mathcal{S}^{k-1}}{\bf \hat x}_{\mathcal{S}^{k-1}}\nonumber \\
%\nonumber \\ &=\sum_{\ell \in \mathcal{T}\setminus\mathcal{S}^{k-1} }{\bf a}_{\ell}{\hat x}_{\ell}+{\bf w}, \nonumber \\
&=\sum_{\ell \in \mathcal{T}\setminus\mathcal{S}^{k-1} }{\bf a}_{\ell}{x}_{\ell}+\sum_{i \in \mathcal{S}^{k-1} }{\bf a}_{i}{e}_{i}+{\bf w},
\end{align}
where $e_{i}={\hat x}_{i}-{x}_{i}$ denotes the estimation error by the orthogonal projection. To identify the support element, the MAP-OMP algorithm performs hypothesis testing by computing the correlation value $z_n^k=\frac{{\bf a}_n^T{\bf r}^{k-1}}{\|{\bf a}_n\|_2}$ as
\begin{align}
\mathcal{H}_0: z_n^k=&\sum_{\ell \in \mathcal{T}\setminus\mathcal{S}^{k-1}}\!\!\frac{{\bf a}_n^T{\bf a}_{\ell}}{\|{\bf a}_n\|_2}{ x}_{\ell} + \sum_{i \in \mathcal{S}^{k-1}}\!\!\frac{{\bf a}_n^T{\bf a}_{i}}{\|{\bf a}_n\|_2}{e}_{i}  + \frac{{\bf a}_n^T{\bf w}}{\|{\bf a}_n\|_2}\nonumber\\
\mathcal{H}_{1}:z_n^k=& \|{\bf a}_n\|_2x_n +\!\!\!\!\!\sum_{\ell \in \mathcal{T}\setminus\{\mathcal{S}^{k-1}\cup \{n\}\}}\!\!\frac{{\bf a}_n^T{\bf a}_{\ell}}{\|{\bf a}_n\|_2}{ x}_{\ell} \nonumber\\&+  \sum_{i \in \mathcal{S}^{k-1}}\!\!\frac{{\bf a}_n^T{\bf a}_{i}}{\|{\bf a}_n\|_2}{e}_{i} + \frac{{\bf a}_n^T{\bf w}}{\|{\bf a}_n\|_2},
\end{align}
where $x_{n}$ is distributed as $f_{x_n}(u)$. The exact characterization of the distribution for $z_n^k$ under the null hypothesis is challenging, as it highly depends on the signal distribution $f_{x_n}(u)$. To facilitate simplified calculations, the distribution of $z_n^k$ is approximated using Gaussian distribution with the first and the second order moments matching. From Lemma \ref{lemma1}, recall that $\frac{{\bf a}_n^T{\bf a}_{\ell}}{\|{\bf a}_n\|_2}$ and $\frac{{\bf a}_n^T{\bf w}}{\|{\bf a}_n\|_2}$ are distributed as $\mathcal{N}(0,\frac{K-(k-1)}{M})$ and $\mathcal{N}(0,\sigma_w^2)$. Furthermore, since $\mathbb{E}[x_{\ell}]=\mu$ and $\mathbb{E}[x_{\ell}^2]=\sigma_x^2$ for $\ell\in\mathcal{T}$, the first and second moments of $z_n^k$ are
\begin{align}
\mathbb{E}\left [ z_n^k \mid x_n=0\right]&=\!\!\!\sum_{\ell \in \mathcal{T}\setminus\mathcal{S}^{k-1}}\mathbb{E}\left [ \frac{{\bf a}_n^T{\bf a}_{\ell}}{\|{\bf a}_n\|_2} \right]\mathbb{E}[x_{\ell}] \nonumber \\&+ \sum_{i \in \mathcal{S}^{k-1}}\!\!\mathbb{E}\left [  \frac{{\bf a}_n^T{\bf a}_{i}}{\|{\bf a}_n\|_2}\right]\mathbb{E}[{e}_{i}]  +\mathbb{E}\left[\frac{{\bf a}_n^T{\bf w}}{\|{\bf a}_n\|_2}\right]\nonumber \\&=0
\end{align} and
\begin{align}
&\mathbb{E}\left [ (z_n^k)^2 \mid x_n=0\right]
= \!\!\sum_{\ell \in \mathcal{T}\!\setminus\!\mathcal{S}^{k\!-\!1}} \mathbb{E}\!\left [ \!\left(\! \frac{{\bf a}_n^T{\bf a}_{\ell}}{\|{\bf a}_n\|_2} \!\right)^2\! \right]\!\mathbb{E}\!\left[\!x_{\ell}^2\right] \nonumber \\
&+ \sum_{i \in \mathcal{S}^{k\!-\!1}}\!\!\mathbb{E}\!\left [ \! \left(\frac{{\bf a}_n^T{\bf a}_{i}}{\|{\bf a}_n\|_2}\!\right)^{\!2}\!\right]\!\mathbb{E}[{e}_{i}^2] \!+\!\mathbb{E}\!\left[\!\left(\!\frac{{\bf a}_n^T{\bf w}}{\|{\bf a}_n\|_2}\!\right)^{\!2}\!\right] \nonumber \\
\!\!\!& \!\!=\frac{(K\!-\!k\!+\!1)\sigma_x^2}{M}\!+\!\frac{(k\!-\!1)}{M}\frac{\sigma^2_x(K-k+1)+{\tilde \sigma}^2_w}{M-k-1}+\frac{\tilde \sigma_w^2}{M} \nonumber \\
\!\!\!& \!\!=\frac{(K\!-\!k\!+\!1)\sigma_x^2+\tilde \sigma_w^2}{M} \left(1+\frac{k\!-\!1}{M-k-1}\right),
\end{align}
where $\mathbb{E}[{e}_{i}^2]= \frac{\sigma^2_x(K-k+1)+{\tilde \sigma}^2_w}{M-k-1}$  from Lemma \ref{lemma4}. 
Accordingly, the approximated distribution of $z_n^k$ is given by
\begin{align}
\mathbb{P}\left(z_{n}^k|x_{n}=0\right)&\simeq\frac{1}{\tilde\sigma_{0}\sqrt{2\pi}}\exp\left(-\frac{|z_{n}^k|^2}{2\tilde\sigma_{0}^2}\right), \label{eq:PDF_null_sig_nonbinary}
\end{align}
where $\tilde\sigma_{0}=\sqrt{\frac{(K-k+1)\sigma_x^2+\tilde \sigma_w^2}{M} \left(1+\frac{k-1}{M-k-1}\right)}$. Similarly, conditioning the hypothesis of $x_n=u$, the approximated distribution of $z^k_n$ is given by
\begin{align}
\mathbb{P}\left(z_{n}^k|x_{n}=u\right)&\simeq\frac{1}{\tilde\sigma_{1}\sqrt{2\pi}}\exp\left(-\frac{|z_{n}^k-u|^2}{2\tilde\sigma_{1}^2}\right),\label{eq:PDF_notnull_sig_nonbinary}
\end{align}
where $\tilde\sigma_{1}=\sqrt{\frac{(K-k)\sigma_x^2+\tilde \sigma_w^2}{M}\!+\!\frac{(k\!-\!1)}{M}\frac{\sigma^2_x(K-k+1)+{\tilde \sigma}^2_w}{M-k-1}}$. Utilizing the approximated distributions in \eqref{eq:PDF_null_sig_nonbinary} and \eqref{eq:PDF_notnull_sig_nonbinary}, the log-MAP ratio is obtained by marginalizing with respect to the distribution $f_{x_n}(u)$, namely,
\begin{align}
\Lambda\left(z_{n}^k\right)&\simeq \ln\left(\frac{\int_{-\infty}^{\infty}\mathbb{P}\left(z_{n}^k|x_{n}=u\right)f_{x_n}(u) {\rm d}u   }{\mathbb{P}\left(z_{n}^k|x_{n}=0\right)}\right) \!+\! \ln\!\left(\!\frac{ K}{N\!-\!K}\!\!\right). \label{eq:MAPratio_general}
\end{align}  
Therefore, the proposed MAP support detection for the non-binary signal is to select the support index such that 
\begin{align}
&\arg\max_{n\in[1:N]} \Lambda\left(z_{n}^k\right) \nonumber\\ &\simeq \arg\max_{n\in[1:N]} \ln\left(\frac{\int_{-\infty}^{\infty}\mathbb{P}\left(z_{n}^k|x_{n}=u\right)f_{x_n}(u) {\rm d}u   }{\mathbb{P}\left(z_{n}^k|x_{n}=0\right)}\right). \label{eq:MAPratio_general}
\end{align}  

To provide a more transparent interpretation of the expression in (\ref{eq:MAPratio_general}), we consider the following three cases of interest.\\

\textbf{Example 1} (Uniformly Distributed Signal): One basic case is the scenario where the elements of the transmit signal are drawn from a uniform distribution between 0 and 1, i.e., $f_{x_n}(u)=1$ for $0\leq u\leq 1$. Then, the MAP ratio expression in (\ref{eq:MAPratio_general}) becomes 
\begin{align}
\Lambda_{U}\left(z_{n}^k\right)&\simeq \ln\!\!\left(\! \frac{\frac{\tilde\sigma_1\sqrt{\pi}}{2}{\rm Erf}\left[\!\left(\frac{1-z_n^k}{\tilde\sigma_1}\right)\!+\!{\rm Erf}\left(\frac{z_n^k}{\tilde\sigma_1}\right)\!\right]}{\frac{1}{\sqrt{2\pi}\tilde\sigma_0}\exp\left(-\frac{(z_n^k)^2}{2\tilde\sigma_0^2}\right)}\!\right),  \label{eq:UniformSignal}
\end{align}  
where ${\rm Erf}(x)=\frac{2}{\sqrt{\pi}}\int_{0}^xe^{-t^2} {\rm d}t.$

\textbf{Example 2} (Sparse Signal with Finite Alphabet): Another popular case of interest is one where the non-zero entry of ${\bf x}$ is uniformly selected from the elements of a finite set of alphabet $\mathcal{C}=\{c_1,\ldots,c_Q\}$ as considered in \cite{Tian2009,Vishwanath2013}. For example, each pixel of a bitmap image file is capable of storing 8 different colors when the 3-bit per pixel (8bpp) format is used. In this application, the finite set can be given as $\mathcal{C}=\left\{0,1,\ldots,7\right\}$. In this case, the log-MAP is computed as follows:
\begin{align}
\Lambda_{C}\left(z_{n}^k\right)&\simeq \ln\!\!\left(\frac{ \sum_{i=1}^{|\mathcal{C}|}\mathbb{P}\left(z_{n}^k|x_{n}=c_{i}\right) \mathbb{P}[x_n=c_{i}] }{\mathbb{P}\left(z_{n}^k|x_{n}=0\right)}\right)  \nonumber\\
&=\ln\left(\frac{ \frac{1}{|\mathcal{C}|} \sum_{i=1}^{|\mathcal{C}|} \frac{1}{\sqrt{2\pi}\tilde\sigma_1 }\exp\left(\frac{-(z_n^k-c_{i})^2}{2\tilde\sigma_1^2}\right)}{\frac{1}{\sqrt{2\pi}\tilde\sigma_0}\exp\left(-\frac{(z_n^k)^2}{2\tilde\sigma_0^2}\right)}\right).\label{eq:MAPratio_PAM}
\end{align}
It is worth noting that when $|\mathcal{C}|=1$, this MAP ratio approximation in \eqref{eq:MAPratio_PAM} recovers the exact MAP ratio for the binary signal case given in \eqref{eq:MAP_Ratio}. 

\textbf{Example 3} (Gaussian Signal):  Assuming $f_{x_n}(u)=\frac{1}{\sigma_x\sqrt{2\pi}}e^{-\frac{(u-\mu)^2}{2\sigma_x^2}}$, the log-MAP ratio simplifies to
\begin{align}
\Lambda_{G}\left(z_{n}^k\right)&\simeq \ln\!\!\left(\frac{\int_{-\infty}^{\infty}\mathbb{P}\left(z_{n}^k|x_{n}=u\right)\frac{1}{\sigma_x\sqrt{2\pi}}e^{-\frac{(u-\mu)^2}{2\sigma_x^2}}{\rm d}u   }{\mathbb{P}\left(z_{n}^k|x_{n}=0\right)}\right)  \nonumber\\
&=\ln\left( \frac{  \frac{1}{2\pi\sqrt{ \sigma_x^2+\tilde\sigma_1^2}}\exp\left(\frac{-(z_n^k-\mu)^2}{2(\sigma_x^2+\tilde\sigma_1^2)}\right)}{\frac{1}{\sqrt{2\pi}\tilde\sigma_0}\exp\left(-\frac{(z_n^k)^2}{2\tilde\sigma_0^2}\right)}\right)  \nonumber \\ 
\!\!\!&\!\!=\frac{(z_n^k)^2}{2 \tilde\sigma_0^2 }\!-\!\frac{(z_n^k-\mu)^2}{2(\sigma_x^2\!+\!\tilde\sigma_1^2)}+ \ln\left(\frac{\tilde\sigma_0}{\sqrt{2\pi( \sigma_x^2\!+\!\tilde\sigma_1^2)}}\right)\!\!.\label{eq:MAPratio_Gaussian}
\end{align}
In the case in which the signal is distributed as zero-mean Gaussian, i.e., $\mathbb{E}[x_{\ell}]=0$, the algorithm selects the index that maximizes $\Lambda_{G}\left(z_{n}^k\right)=(z_n^k)^2 \left(\frac{1}{2 \tilde\sigma_0^2 }\!-\!\frac{1}{2(\sigma_x^2\!+\!\tilde\sigma_1^2)}\right)$, which is the same selection criterion used in the conventional OMP algorithm; thereby, there is no benefits of using the proposed method compared to the OMP algorithm.
Whereas, when the Gaussian signal has a non-zero mean value, i.e., $\mathbb{E}[x_{\ell}]\neq 0$, the proposed algorithm provides a better support detection probability than that of the conventional OMP algorithm.

Using theses approximated log-MAP ratio functions in the examples, we provide the MAP-OMP algorithm, which is summarized in Table \ref{table1-1}. The key difference with the MAP-MP algorithm for the binary signal is that MAP-OMP computes the MAP ratio differently depending on the sparse signal distribution. Furthermore, the algorithm estimates the sparse signal using a least square solution in each iteration similar to the conventional OMP algorithm.

\begin{table}
\begin{center}\vspace{-0.2cm}
\caption{MAP-OMP Algorithm}\vspace{-0.2cm}
\begin{tabular}{|l|}\hline\hline
1)  Initialization:\\
\hspace{5mm} $k:=0$, $\mathbf{\hat x}^{0} = {\bf 0}$ \\
\hspace{5mm} ${\bf r}^0:={\bf y}$ (the current residual)\\
\hspace{5mm}  $\mathcal{S}^0:=\{\emptyset\}.$\\
\hline
2) Repeat until a stopping criterion is met\\
\hspace{5mm} i) $k:=k+1$. \\
\hspace{5mm} ii) Compute the current proxy: \\
\hspace{10mm} $z^k_n = \frac{{\bf a}^T_n\mathbf{r}^{k-1}}{\|{\bf a}_n\|_2}$ for $n\in[1:N]$.\\
\hspace{5mm} iii) Select the largest index of MAP ratio:\\
\hspace{10mm} $J^k=: \arg\max_{n} \left\{\Lambda_{d}(z_{n}^k)\right\}$ for $d\in\{U,C,G\}$.\\
\hspace{5mm} iv) Merge the support set:\\
\hspace{10mm} $\mathcal{S}^k=\mathcal{S}^{k-1}\cup J^k$.\\
\hspace{5mm} v) Update sparse signal:\\ 
\hspace{10mm} 
${\bf \hat x}^k_{\mid \mathcal{S}^k}:=\arg \min_{{\bf x}}\|{\bf \Phi}_{\mid\mathcal{S}^k}{\bf x}-{\bf y}\|_2$.\\
\hspace{5mm} vi) Update the residual for next round:\\
\hspace{10mm} ${\bf r}^k:={\bf y}-{\bf \Phi}_{\mid \mathcal{S}^k}{\bf \hat x}^k_{\mid \mathcal{S}^k}$.\\
\hline\hline \end{tabular} \label{table1-1}  
\end{center}\vspace{-0.2cm}
\end{table}

\vspace{-0.2cm}
\section{Extension to the Other Greedy Algorithms}
One advantage of the proposed MAP support detection method is, indeed, directly applicable to many other greedy sparse signal recovery algorithms. In this section, we provide a set of greedy sparse signal recovery algorithms that exploit the proposed support detection method. %We refer them to MAP-CoSaMP, MAP-gOMP, and MAP-SP. 
\vspace{-0.2cm}
\subsection{MAP-gOMP}
gOMP \cite{gOMP} is a simple yet effective algorithm that improves the performance of OMP. The key idea of gOMP is the selection of multiple support indices with the largest correlation in magnitude at each iteration; thereby, it reduces the mis-detection probability compared to that of OMP.  Similar to the gOMP algorithm, MAP-gOMP is a greedy algorithm that sequentially finds multiple support indices and estimates the signal representation within a certain number of iterations. The core difference lies in the selection rule of the support indices per iteration. Unlike the gOMP algorithm, MAP-gOMP chooses $L$ support indices with the largest log-MAP ratio values instead of the largest correlation in magnitude. Therefore, in the $k$th iteration, we update the variances of two conditional distributions in \eqref{eq:PDF_null_sig_nonbinary} and \eqref{eq:PDF_notnull_sig_nonbinary} as $\tilde\sigma_{0}^2= \frac{(K-L(k-1))\sigma_x^2+\tilde \sigma_w^2}{M} \left(1+\frac{L(k-1)}{M-L(k-1)-2}\right)$ and ${\tilde \sigma}_1^2=\frac{(K-Lk)\sigma_x^2+\tilde \sigma_w^2}{M}\!+\!\frac{L(k\!-\!1)}{M}\frac{\sigma^2_x(K-L(k-1))+{\tilde \sigma}^2_w}{M-L(k-1)-2}$. The proposed MAP-gOMP is summarized in Table \ref{table2}.

\begin{table}
\center \vspace{-0.2cm}
\caption{MAP-gOMP Algorithm}\vspace{-0.2cm}
\begin{tabular}{|l|}\hline\hline
1)  Initialization:\\
\hspace{5mm} $k:=0$, $\mathbf{\hat x}^{0} = {\bf 0}$ \\
\hspace{5mm} ${\bf r}^0:={\bf y}$ (the current residual)\\
\hspace{5mm}  $\mathcal{S}^0:=\{\emptyset\}$ and $\Omega^0:=\{\emptyset\}$\\
\hline
2) Repeat until a stopping criterion is met\\
\hspace{5mm} i) $k:=k+1$. \\
\hspace{5mm} ii) Compute the current proxy: \\
\hspace{10mm} $z^k_n = \frac{{\bf a}^T_n\mathbf{r}^{k-1}}{ \|{\bf a}_n\|_2}$ for $n\in[1:N]$.\\
\hspace{5mm} iii) Select the $L(\leq \frac{M}{K})$ largest indices of MAP ratio:\\
\hspace{10mm} $\Omega^k=: \arg\max_{|\Omega^k|=L} \left\{\Lambda_d(z_{n}^k)\right\}$ for $d\in\{U,C,G\}$.\\
\hspace{5mm} iv) Merge the support set:\\
\hspace{10mm} $\mathcal{S}^k=\mathcal{S}^{k-1}\cup \Omega^k$.\\
\hspace{5mm} v) Perform a least-squares signal estimation:\\ 
\hspace{10mm} ${\bf \hat x}^k_{\mid \mathcal{S}^k}:=\arg \min_{{\bf x}}\|{\bf \Phi}_{\mid\mathcal{S}^k}{\bf x}-{\bf y}\|_2$.\\
\hspace{5mm} vi) Update the residual for next round:\\
\hspace{10mm} ${\bf r}^k:={\bf y}-{\bf \Phi}_{\mid \mathcal{S}^k}{\bf \hat x}^k_{\mid \mathcal{S}^k}$.\\
\hline\hline
\end{tabular}\label{table2}\vspace{-0.2cm}
\end{table}\vspace{-0.2cm}

\subsection{MAP-CoSaMP}
CoSaMP is an effective iterative sparse signal recovery algorithm \cite{Needell2010_CoSaMP}. It was shown to  yield the same sparse signal recovery performance guarantees as $\ell_1$-norm minimization even with less computational complexity. The main idea of CoSaMP is that, in the first step, it estimates a large support set with $L$ largest correlation values in magnitude and obtains a least square solution based on it, where $L$ is typically chosen between $K\leq L \leq 2K$. In the next step, the algorithm reduces the cardinality of the support set back to the desired sparsity level of $K$ using pruning, and acquires a sparse solution again based on the reduced support.

We modify this algorithm by incorporating the proposed support detection technique. Unlike the conventional CoSaMP algorithm, MAP-CoSaMP adds $2K$ support candidates with $2K$ largest MAP ratio values to the support set $\mathcal{S}^k$ per iteration. Once the least square solution is obtained based on the corresponding support set $\mathbf{\hat x}_{\mid \mathcal{S}^{k}} = \mathbf{\Phi}_{\mid \mathcal{S}^{k}}^\dagger \mathbf{y}$, an approximation to the signal is updated by selecting the $K$ largest coordinates using pruning. Finally, the residual is updated using the approximated signal estimate. The algorithm is described in Table \ref{table3}. The computational complexity order of the proposed MAP-CoSaMP is the same with that of the original CoSaMP algorithm \cite{Needell2010_CoSaMP}. We refer \cite{Needell2010_CoSaMP,Needell_Tropp_2008} for the reader who are interested in the computational complexity analysis of CoSaMP.

\begin{table}\center\vspace{-0.2cm}
\caption{MAP-CoSaMP Algorithm}\vspace{-0.2cm}
\begin{tabular}{|l|}
\hline\hline
1)  Initialization:\\
\hspace{5mm} $k:=0$, $\mathbf{\hat x}^{0} = {\bf 0}$ \\
\hspace{5mm} ${\bf r}^0:={\bf y}$ (the current residual)\\
\hspace{5mm}  $\mathcal{S}^0:=\{\emptyset\}$ and $\Omega^0:=\{\emptyset\}$\\
\hline
2) Repeat until a stopping criterion is met\\
\hspace{5mm} i) Compute the current proxy: \\
\hspace{10mm} $z^k_n =  \frac{{\bf a}^T_n\mathbf{r}^{k-1}}{\|{\bf a}_n\|_2}$ for $n\in[1:N]$.\\
\hspace{5mm} iii) Select the $2K$ largest indices of MAP ratio:\\
\hspace{10mm} $\Omega^k=: \arg\max_{|\Omega^k|=2K} \left\{\Lambda_d(z_{n}^k)\right\}$ for $d\in\{U,C,G\}$.\\
\hspace{5mm} iv) Merge the support set:\\
\hspace{10mm} $\mathcal{S}^k=\mathcal{S}^{k-1}\cup \Omega^k$.\\
\hspace{5mm} iv) Perform a least-squares signal estimation:\\ 
\hspace{10mm} $\mathbf{\hat x}_{\mid \mathcal{S}^{k}} =:\arg \min_{{\bf x}}\|{\bf \Phi}_{\mid\mathcal{S}^k}{\bf x}-{\bf y}\|_2$, $\mathbf{\hat x}_{\mid {\mathcal{S}^{k}}^c} = 0$.\\
\hspace{5mm} v) Prune $\mathbf{\hat x}^k$:\\
\hspace{10mm}  ${\mathcal G}=: \arg\max_{|{\mathcal G}|=K} \left\{|{\bf \hat x}^k|\right\}$,\\
\hspace{5mm} vi) Update the residual for next round:\\
\hspace{10mm} $\mathbf{r}^k = \mathbf{y}-\mathbf{\Phi}_{\mid \mathcal{G}} \mathbf{\hat x}^k_{\mid \mathcal{G}}$.\\
\hline\hline
\end{tabular}
\label{table3}\vspace{-0.2cm}
\end{table}

\vspace{-0.2cm}
\subsection{MAP-SP Algorithm}
 SP is a two-step iterative algorithm for sparse recovery \cite{Dai2009_SP}. Similar to CoSaMP, the SP algorithm identifies the current estimate of support set by greedily adding multiple indices with the largest correlation in magnitude.The main difference between CoSaMP and SP lies in the second step. While CoSaMP applies a pruning technique using the estimated sparse signal in the first stage to maintain the required sparsity level without performing the second least-square estimation. Whereas, the SP algorithm updates the sparse solution by solving a least square problem based on the reduced support in the second stage.

\begin{table}
\center\vspace{-0.2cm}
\caption{MAP-SP Algorithm}\vspace{-0.2cm}
\begin{tabular}{|l|}
\hline\hline
1)  Initialization:\\
\hspace{5mm} $k:=0$, $\mathbf{\hat x}^{0} = {\bf 0}$ \\
\hspace{5mm} ${\bf r}^0:={\bf y}$ (the current residual)\\
\hspace{5mm}  $\mathcal{S}^0:=\{\emptyset\}$ and $\Omega^0:=\{\emptyset\}$\\
\hline
2) Repeat until a stopping criterion is met\\
\hspace{5mm} i) Compute the current proxy: \\
\hspace{10mm} $z^k_n = \frac{{\bf a}^T_n\mathbf{r}^{k-1}}{\|{\bf a}_n\|_2}$ for $n\in[1:N]$.\\
\hspace{5mm} iii) Select the $K$ largest indices of MAP ratio:\\
\hspace{10mm} $\Omega^k=: \arg\max_{|\Omega^k|=K} \left\{\Lambda_{d}(z_{n}^k)\right\}$ for $d\in\{U,C,G\}$.\\
\hspace{5mm} iv) Merge the support set:\\
\hspace{10mm} $\mathcal{S}^k=\mathcal{S}^{k-1}\cup \Omega^k$.\\
\hspace{5mm} iv) Perform a least-squares signal estimation:\\ 
\hspace{10mm} ${\bf b}^k:=\arg \min_{{\bf b}}\|{\bf \Phi}_{\mid\mathcal{S}^k}{\bf b}-{\bf y}\|_2$\\
\hspace{5mm} v) Select the $K$ largest index in $\mathbf{\hat x}^k$:\\
\hspace{10mm}  ${\mathcal G}=: \arg\max_{|{\mathcal G}|=K} \left\{|{\bf b}^k|\right\}$\\
\hspace{5mm} vi) Perform a least-squares signal estimation using the updated ${\mathcal G}$:\\ 
\hspace{10mm} ${\bf \hat x}^k_{\mid \mathcal{G}}:=\arg \min_{{\bf x}}\|{\bf \Phi}_{\mid\mathcal{G}}{\bf x}-{\bf y}\|_2$.\\
\hspace{5mm} vii) Update the residual for next round:\\
\hspace{10mm} $\mathbf{r}^k = \mathbf{y}-\mathbf{\Phi}_{\mid \mathcal{G}} \mathbf{\hat x}^k_{\mid \mathcal{G}}$.\\
\hline\hline
\end{tabular}\label{table4}\vspace{-0.2cm}
\end{table}

Applying the proposed MAP support detection method, we modify this algorithm by changing the support set identification stage. The proposed MAP-SP algorithm selects $2K$ support indices with the largest MAP ratio values in each iteration. The modified algorithm is summarized in Table \ref{table4}. Since the log-MAP ratio computation does not increase the computational complexity order, the proposed algorithm can be implemented with $\mathcal{O}(MNK)$, which is comparable to that of the SP algorithm in \cite{Dai2009_SP}.

\section{Numerical Results}

We provide empirical recovery performance of the proposed algorithms by means of simulations. We evaluate the empirical frequency (cumulative density function) of exact reconstruction for the proposed algorithms in both noise and noiseless cases and compare them with the conventional algorithms. In our simulation, we generate $M\times N$ ($M=128$ and $N=256$) sensing matrix whose elements are drawn from IID Gaussian distribution $\mathcal{N}(0, \frac{1}{M})$. Furthermore, we consider $K$-sparse vector ${\bf x}$ whose support is uniformly distributed. Each non-zero element of ${\bf x}$ is one for the binary signal and is randomly selected from $[0,1]$ for the uniform signal. To obtain the empirical frequency of exact reconstruction, we perform 1,000 independent trials for each algorithm. For each trial, we perform iterations until the stopping criterion $\|{\bf x}-{\bf \hat{x}}\|_2^2\leq 10^{-12}$ is satisfied except for gOMP and MAP-gOMP. For gOMP and MAP-gOMP, we perform $\min\left(K,\left \lfloor{\frac{M}{L}}\right \rfloor \right)$ number of iterations in each trial, where $L=2$. To obtain the performance of BP, we use the CVX tool that can be executed in \textsf{MATLAB} \cite{cvx}. 

\begin{figure}
\centering \label{Fig1}\vspace{-0.2cm}
\includegraphics[width=3.75in]{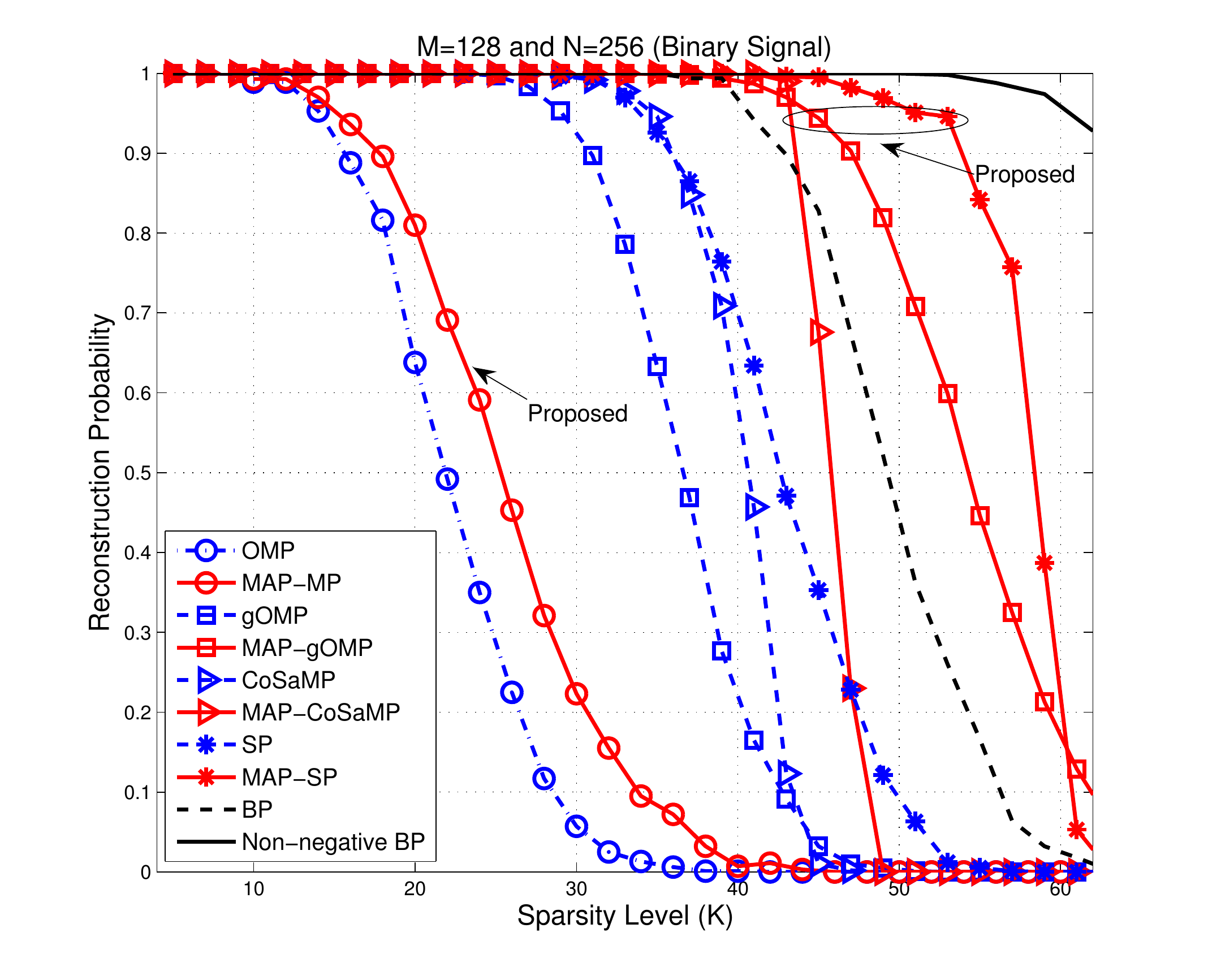} \vspace{-0.8cm}
\caption{Performance comparison of perfect reconstruction probability for the binary signal with noise-free measurements.} \label{fig:1} \vspace{-0.3cm}
\end{figure}

\begin{figure}
\centering \label{Fig2}
\includegraphics[width=3.75in]{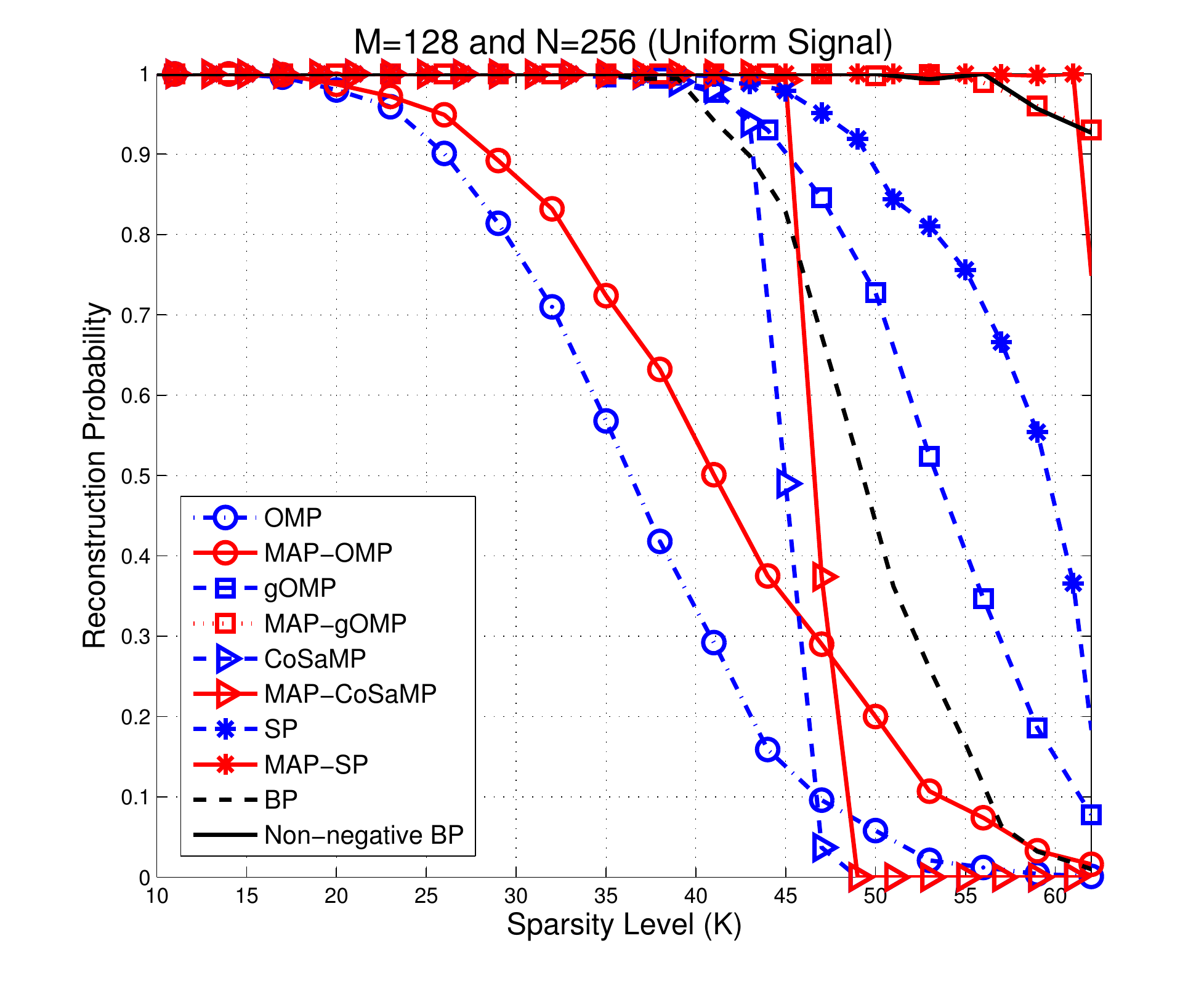} \vspace{-0.8cm}
\caption{Performance comparison of perfect reconstruction probability for the signal whose non-zero element is uniformly distributed between 0 and 1, i.e., $x_i\sim {\rm Uni}[0,1]$ with noise-free measurements.} \label{fig:2} \vspace{-0.3cm}
\end{figure}

Fig. \ref{fig:1} illustrates the reconstruction probability  performance of a binary sparse signal with noise-free measurements as a function of the sparsity level $K$ of the signal. The simulation results reveal that the proposed algorithms improve the reconstruction probability performance significantly compared to those of the existing algorithms. For example, the proposed MAP-gOMP recovers the binary sparse signal with more than 90 $\%$ probability up to a sparsity level of $42$. Whereas, the conventional gOMP is able to reconstruct the signal only up to a sparsity level $31$ under the same reconstruction probability constraint. Furthermore, the proposed MAP-gOMP, MAP-CoSaMP, and MAP-SP algorithms outperform  BP, i.e., a linear programing technique, for the binary signal reconstruction. A non-negative BP algorithm that solves the $\ell_1$-minimization problem with an additional non-negative constraint in ${\bf x}$, however, provides a better performance than the proposed algorithms at the expense of a more computational complexity.

Fig. \ref{fig:2} shows the reconstruction probability of a sparse signal whose non-zero element is uniformly distributed between 0 and 1, i.e., $x_i\sim {\rm Uni}[0,1]$. We use the MAP ratio function in \eqref{eq:UniformSignal}  for the simulations. Similar to the binary signal case, it is no wonder that the proposed MAP-gOMP, MAP-CoSaMP, and MAP-SP algorithms outperform than the existing sparse recovery algorithms by considerably reducing the mis-detection probability of supports. In particular, MAP-gOMP and MAP-SP are able to recover the signal with more than 95 $\%$ probability up to a sparsity level of $60$, which is close to the maximum sparsity level ($\frac{M}{2}=64$) that can be recovered with a unique solution guarantee. In particular, MAP-SP outperforms than the non-negative BP algorithm. 

We consider now a binary sparse image recovery example. As illustrated in Fig. \ref{fig:3} (the left-top figure), a binary sparse image with 37$\times$37-pixel size is considered for the experiment. Applying linear random projection matrix ${\bf \Phi} \in{\mathbb R}^{685\times 1369}$ whose elements are drawn from $\mathcal{N}(0, \frac{1}{685})$, we compress the binary image. As shown in Fig. \ref{fig:3}, when the noise-free measurements are used for image (supports) reconstruction, we observe that the proposed MAP-gOMP and MAP-SP algorithms for the support recovery outperform than the other existing algorithms, which agrees with the result shown in Fig. \ref{fig:1}. We add Gaussian noise with zero mean and variance $\sigma_w^2=0.005$. In this case, as depicted in Fig. \ref{fig:4}, the proposed MAP-SP method is able to recover the image almost perfectly even in the presence of noise. Whereas, the image reconstruction performance of the proposed MAP-gOMP algorithm is degraded compared  to the case of noise-free, which exhibits the noise sensitivity of the algorithm.

As can be seen in Table \ref{Table7}, the proposed algorithms achieve significant speedup compared to the existing algorithms in both the noise-free and noisy measurements cases. These speedup gains are mainly due to the fact that the proposed algorithms identify the true support set with small number of iterations, leading to the faster convergence rates than those of the existing algorithms. In particular, the runtimes of MAP-SP ($\simeq$ 0.21 sec) under the noise-free measurements speed up $157$ times than that of BP ($\simeq 33.22$ sec).

To provide the insight on how the performance of the proposed algorithm decreases as the noise variance increases for  given $K$, $M$, and $N$, we plot normalized mean squared error (NMSE) of the proposed algorithms as a function of signal-to-noise ratio ${\rm SNR}=\frac{\|{\bf \Phi}{\bf x}\|_2^2}{\sigma_w^2}$, which is defined as
\begin{align}
{\rm NMSE}=10\log_{10}\left(\frac{1}{T}\sum_{i=1}^T\frac{\|{\bf \hat x}^i-{\bf x}^i\|_2^2}{\|{\bf x}^i\|_2^2}\right),
\end{align}
where $T$ is the number of trails and the subscript $i$ represents the trial number. In each random trial, a random gaussian matrix ${\bf \Phi}\in\mathbb{R}^{128\times 256}$ is generated, and the non-zero elements in ${\bf x}$ are generated as Gaussian random variables with mean one and variance $\frac{1}{M^2}=\frac{1}{128^2}$. We assume that sparsity level $K=40$. For this noise case, we slightly modify the least-square signal estimator used in each algorithm such that
 ${\bf \hat x}^k_{\mathcal{S}^k}=\left({\bf \Phi}_{|\mathcal{S}^k}^T{\bf \Phi}_{|\mathcal{S}^k}+\frac{1}{{\rm SNR}}{\bf I} \right)^{-1}{\bf \Phi}_{|\mathcal{S}^k}{\bf y}$. 

As illustrated in Fig. \ref{fig:5}, the algorithms using the proposed MAP support detection method outperform the conventional sparse recovery algorithms. This reveals that the proposed algorithms are robust to measurement noise. Interestingly, the proposed algorithms including MAP-gOMP ($L=2)$ and MAP-SP exhibit a better NMSE performance compared to that of FBMP in \cite{Schniter2009}.

\begin{figure}
\centering \vspace{-0.3cm}
\includegraphics[width=3.5in]{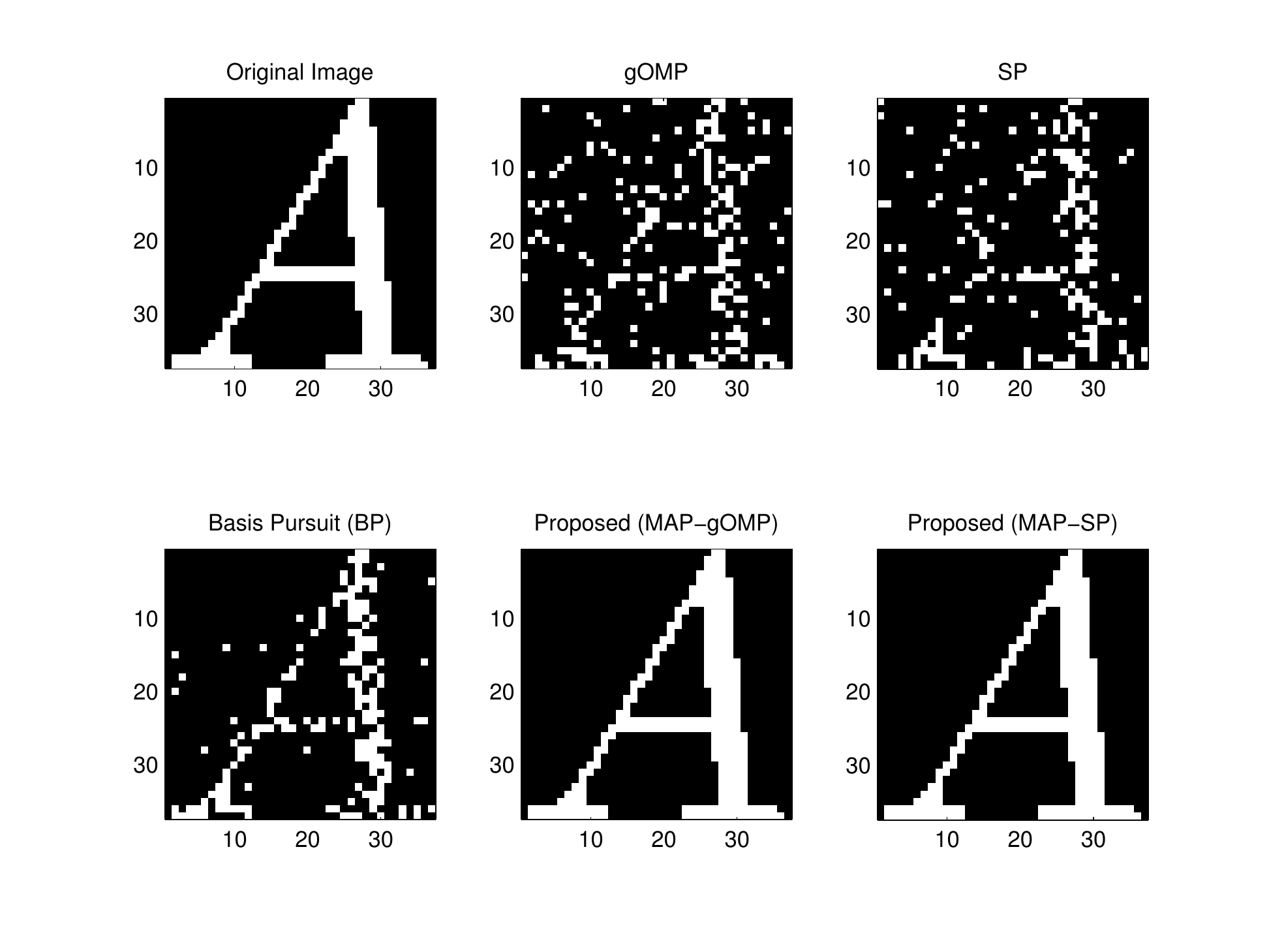} \vspace{-0.9cm}\caption{Support recovery performance comparison of the sparse binary image reconstruction with the compressed and noise-free measurements, i.e., $\sigma_w^2=0$. All different algorithms use the same random linear projection matrix for image reconstruction.} \label{fig:3} \vspace{-0.3cm}
\end{figure}

\section{Conclusion}
We have presented a new support detection technique based on a MAP criterion for greedy sparse signal recovery algorithms. Using this method, we have proposed a set of greedy sparse signal recovery algorithms and established a theoretical signal recovery guarantee for a particular case. One major implication is that the joint use the distributions of sensing matrix, sparse signal, and noise in support identification offers a tremendous recovery performance improvement over previous support detection approaches that ignore such statistical information. Our numerical results demonstrate that the greedy algorithms with highly reliable support detection provide significantly better sparse recovery performance than the linear programming approach. 

An interesting direction for future study would be to explore the statistical guarantees of the proposed MAP-gOMP, MAP-CoSaMP, and MAP-SP. Another possible research direction is to investigate the greedy algorithms when different statistical distributions of the sensing matrix are used. Furthermore, it would be interesting to apply the proposed support detection principle to improve the sparse signal reconstruction method in \cite{Wang_Yin}.
%
 %This suggests that greedy algorithms with highly reliable support detection may be better, faster and easier to implement than basis pursuit via linear . 

\begin{figure}
\centering \vspace{-0.2cm}
\includegraphics[width=3.6
in]{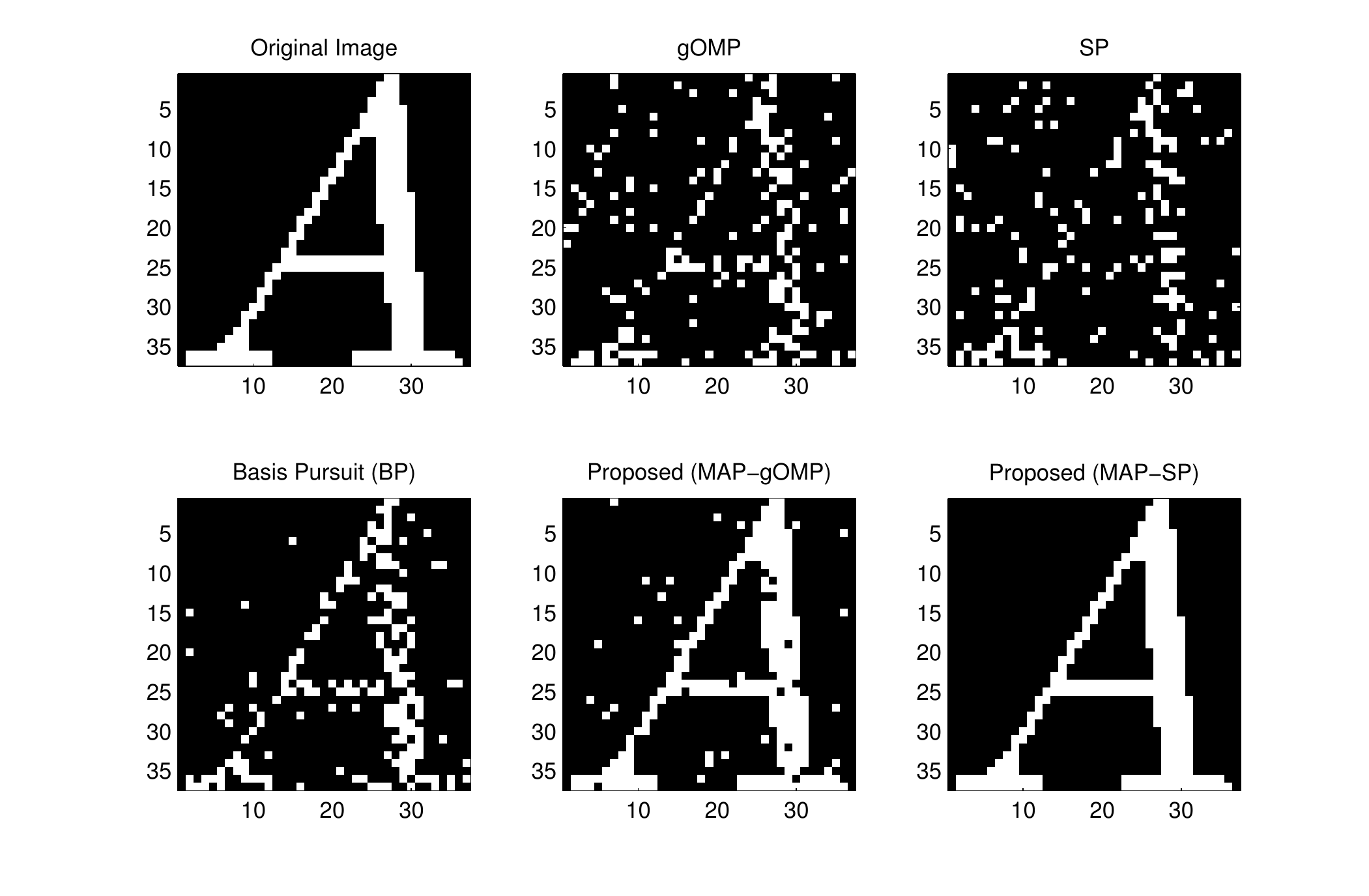} \vspace{-0.7cm}
\caption{Support recovery performance comparison of the sparse binary image reconstruction with the compressed and noisy measurements, where $\sigma_w^2=0.005$, equivalently  ${\tilde \sigma}_w^2=M\sigma_w^2= 3.425$.} \label{fig:4}  \vspace{-0.2cm}
\end{figure}

\begin{table}
\begin{center}\vspace{-0.2cm}
  \caption{Algorithm Runtimes } \vspace{-0.4cm}
\begin{tabular}{c|c|c||c|c} \hline
     &\!\!\!\!\!  Runtimes (Sec) \!\!\!\!\!\! & Speedup & Runtimes   & Speedup   \\
Algorithms &  $\sigma_w^2=0$  & $\sigma_w^2=0$ & $\sigma_w^2=0.005$   & $\sigma_w^2=0.005$   \\ \hline\hline
   gOMP  &  5.03 &  6.6x  & 5.03 &  7.1x\\  
   MAP-gOMP  & 2.26 &  14.6x& 4.81 &  7.4x \\ 
   SP & 12.97 &   2.3x  & 14.5 &   2.5x\\  
   MAP-SP  & 0.21 & 157.8x &  15.1 & 2.4x\\  
   BP  & 33.22 &  baseline & 35.73 &  baseline\\  
\hline
\end{tabular}\label{Table7} \vspace{-0.4cm}
\end{center}
\end{table}

\begin{figure}
\centering \vspace{-0.2cm}
\includegraphics[width=3.8
in]{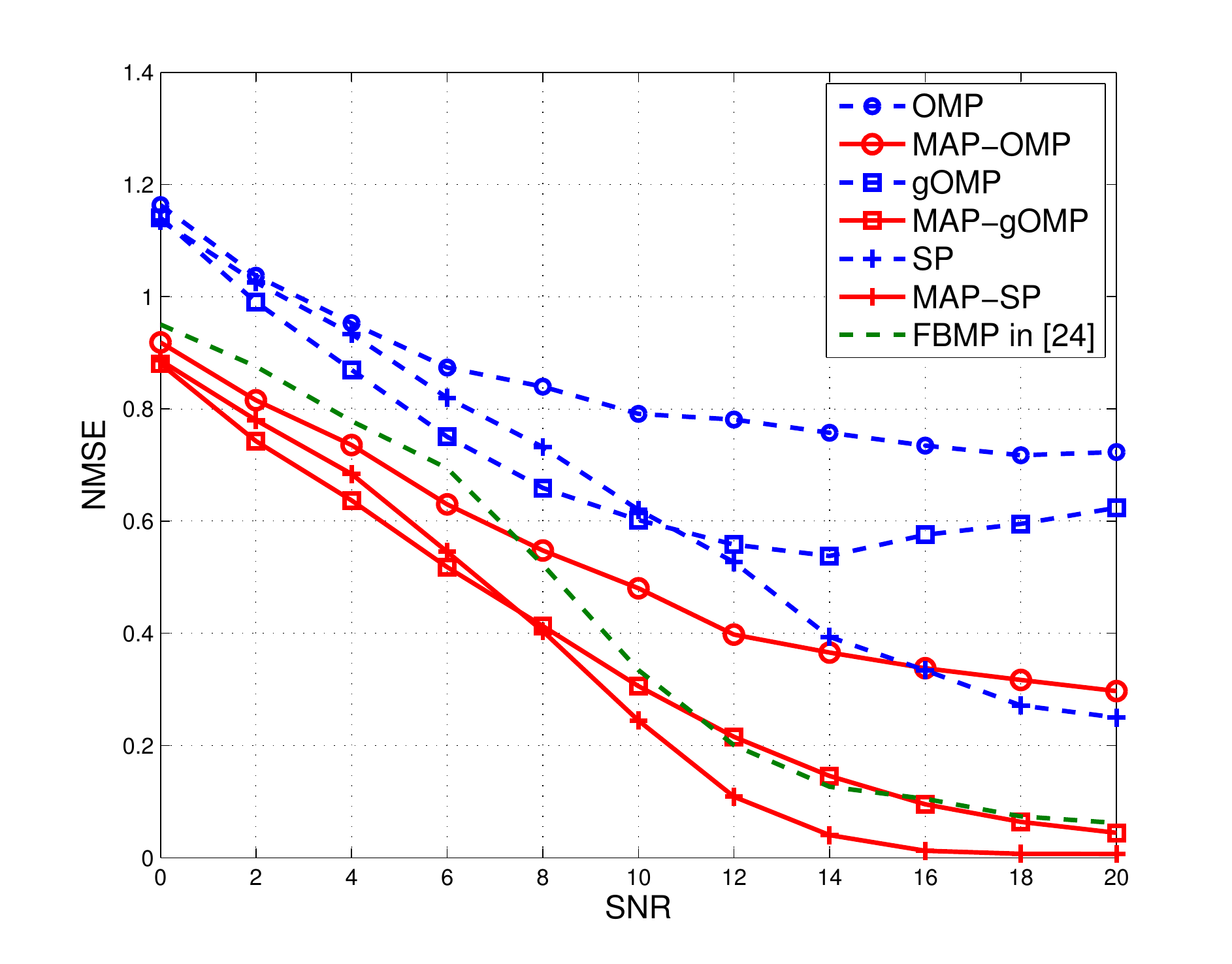} \vspace{-0.7cm}
\caption{The NMSE performance comparison among different sparse signal recovery algorithms when $T=1000$. For FBMP, we use $D = 20$, which is the maximum number of allowable repeated greedy searches.} \label{fig:5}  \vspace{-0.2cm}
\end{figure} 
 
%%%%%%%%%%%%%%%%%%%%%%%%%%%%%
%%%%%%%%%%%%%%%%%%%%%%%%%%%%%
%\appendices 
 
\appendix
\section{Appendices}
%\subsection{First appendix}
%\subsection{Second appendix}

\vspace{-0.1cm}
\subsection{Proof of Lemma \ref{lemma1}}\label{proof:lemma1} \vspace{-0.1cm}
 
Note that the distribution of each atom vector ${\bf a}_n$ is rotationally invariant. This implies that for any unitary matrix ${\bf U}\in\mathbb{R}^{M\times M}$, the distributions of ${\bf U}{\bf a}_n$ and ${\bf a}_n$ are identical. By selecting a unitary matrix ${\bf U}$ so that ${\bf U}{\bf a}_n=[1, 0, \ldots, 0]^T$, we can compute the cumulative distribution function of $\frac{{\bf a}_n^T{\bf a}_{\ell}}{\|{\bf a}_n\|_2}$ as
\begin{align}
\mathbb{P}\left[\frac{{\bf a}_n^T{\bf a}_{\ell}}{\|{\bf a}_n\|_2} \leq x\right] 
&=\mathbb{P}\left[\frac{{\bf a}_n^T}{\|{\bf a}_n\|_2}{\bf U}^T{\bf a}_{\ell}  \leq x\right] 
\nonumber \\
&=\mathbb{P}\left[{\bf a}_{\ell}(1)  \leq x\right]
\end{align}
where ${\bf a}_{\ell}(1)$ denotes the first component of ${\bf a}_{\ell}$. As a result, $\frac{{\bf a}_n^T{\bf a}_{\ell}}{\|{\bf a}_n\|_2}$ is IID Gaussian with zero mean and variance $\frac{1}{M}$.

 \vspace{-0.1cm}
\subsection{Proof of Lemma \ref{lemma2}}\label{proof:lemma2} \vspace{-0.1cm}

Recall that all elements of ${\bf a}_n$ are Gaussian random variables with zero mean and variance $\frac{1}{M}$, and they are mutually independent. Thus,
\begin{align}
\mathbb{P}\left[\|{\bf a}_n\|_2 \leq x \right]&=\mathbb{P}\left[\sqrt{\sum_{m=1}^M \left({\bf a}_n(m)\right)^2}  \leq x\right] \nonumber \\
&= \mathbb{P}\left[\sqrt{\sum_{m=1}^M \frac{\left({\bf a}_n(m)\right)^2}{\frac{1}{M}}}  \leq \sqrt{M}x\right]
\nonumber \\
&\stackrel{(a)}{=}\frac{\gamma\left(\frac{M}{2},\frac{Mx^2}{2}\right)}{\Gamma\left(\frac{M}{2}\right)},
\end{align}
where (a) follows from the fact that $\sqrt{\sum_{m=1}^M \frac{\left({\bf a}_n(m)\right)^2}{1/M}} $ is Chi-distributed with $M$ degrees of freedom, since $\frac{\left({\bf a}_n(m)\right)^2}{1/M}$ is a normal Gaussian with zero mean and unit variance, and $\gamma(s,x)=\int_0^{x}t^{s-1}e^{-t} {\rm d}t$ denotes the lower incomplete gamma function. By taking the derivative with respect to $x$, we obtain the distribution of $\|{\bf a}_n\|_2$ as
\begin{align}
f_{\|{\bf a}_n\|_2}(x)&=\frac{1}{{\rm d}x}\frac{\gamma\left(\frac{M}{2},\frac{Mx^2}{2}\right)}{\Gamma\left(\frac{M}{2}\right)} \nonumber \\ 
&=\frac{2^{1-\frac{M}{2}}M^{\frac{M}{2}}x^{M-1}e^{-\frac{Mx^2}{2}}}{\Gamma\left(\frac{M}{2}\right)}.
\end{align}
Accordingly, the mean of the norm is 
\begin{align}
\mathbb{E}\left[\|{\bf a}_n\|_2\right]&=\int_0^{\infty}\frac{2^{1-\frac{M}{2}}M^{\frac{M}{2}}x^{M}e^{-\frac{Mx^2}{2}}}{\Gamma\left(\frac{M}{2}\right)}  {\rm d} x\nonumber \\&=\sqrt{\frac{2}{M}}\frac{\Gamma\left(\frac{1+M}{2}\right)}{\Gamma\left(\frac{M}{2}\right)},
\end{align}
which completes the proof.

 \vspace{-0.2cm}
\subsection{Proof of Lemma \ref{lemma3}}\label{proof:lemma3}\vspace{-0.1cm}

We commence by computing the probability that the absolute difference between the norm and its average is greater than or equal to a small value $\epsilon$, which is 
\begin{align}
&\mathbb{P}\left[ \left| \|{\bf a}_n\|_2 - \mathbb{E}\left[\|{\bf a}_n\|_2\right] \right|  \geq \epsilon \right] \nonumber \\&=\mathbb{P}\left[ \left| \sqrt{\sum_{m=1}^M \left({\bf a}_n(m)\right)^2}  - \sqrt{\frac{2}{M}}\frac{\Gamma\left(\frac{1+M}{2}\right)}{\Gamma\left(\frac{M}{2}\right)}\right| \geq \epsilon \right] \nonumber \\
&\leq \frac{\mathbb{E}\left[\sum_{m=1}^M \left({\bf a}_n(m)\right)^2\right]-  \frac{2}{M}\left(\frac{\Gamma\left(\frac{1+M}{2}\right)}{\Gamma\left(\frac{M}{2}\right)}\right)^2}{\epsilon}
\nonumber \\
&\leq \frac{\sum_{m=1}^M \mathbb{E}\left[ \left( {\bf a}_n(m)\right)^2\right]-  \frac{2}{M}\left(\frac{\Gamma\left(\frac{1+M}{2}\right)}{\Gamma\left(\frac{M}{2}\right)}\right)^2}{\epsilon}
\nonumber \\
&= \frac{1-  \frac{2}{M}\left(\frac{\Gamma\left(\frac{1+M}{2}\right)}{\Gamma\left(\frac{M}{2}\right)}\right)^2}{\epsilon},
\end{align}
where the inequality follows from Chebyshev's inequality. Since $ \frac{2}{M}\left(\frac{\Gamma\left(\frac{1+M}{2}\right)}{\Gamma\left(\frac{M}{2}\right)}\right)^2$ converges to one as $M$ goes to infinity, we conclude that
\begin{align}
\lim_{M\rightarrow \infty}\mathbb{P}\left[ \left| \|{\bf a}_n\|_2 - \mathbb{E}\left[\|{\bf a}_n\|_2\right] \right|  \geq \epsilon \right] =0
\end{align}
for some $\epsilon >0$. As a result, the norm of each column vector concentrates to the average $\mathbb{E}\left[\|{\bf a}_n\|_2\right]=\sqrt{\frac{2}{M}}\frac{\Gamma\left(\frac{1+M}{2}\right)}{\Gamma\left(\frac{M}{2}\right)}$ and it also converges to one for a large enough $M$ because $\lim_{M\rightarrow \infty}\sqrt{\frac{2}{M}}\frac{\Gamma\left(\frac{1+M}{2}\right)}{\Gamma\left(\frac{M}{2}\right)}=1$. This completes the proof.

\vspace{-0.1cm}
\subsection{Proof of Lemma \ref{lemma4}}\label{proof:lemma4}\vspace{-0.1cm}
 
Let ${\bf P}_{|\mathcal{S}^{k}}=\left({\bf \Phi}_{|\mathcal{S}^{k}}^T{\bf \Phi}_{|\mathcal{S}^{k}}\right)^{-1}{\bf \Phi}_{|\mathcal{S}^{k}}^T$ be a projection matrix to estimate ${\bf x}_{|\mathcal{S}^{k}}$ in the $k$th iteration. Using this, the corresponding non-zero elements are obtained as
\begin{align}
{\hat {\bf x}}_{|\mathcal{S}^{k}}&={\bf P}_{|\mathcal{S}^{k}}{\bf y} \nonumber\\
&={\bf x}_{|\mathcal{S}^{k}}+{\bf P}_{|\mathcal{S}^{k}}\left({\bf \Phi}_{|\mathcal{T}\setminus\mathcal{S}^{k}}{\bf x}_{|\mathcal{T}\setminus\mathcal{S}^{k}} +{\bf w}\right).%\nonumber\\
%&={\bf x}_{|\mathcal{S}^{k}}+{\bf P}_{|\mathcal{S}^{k}}{\bf \tilde w}^k,
\end{align}
Then, the mean of the estimation error is
\begin{align}
\mathbb{E}\left[{\hat {\bf x}}_{|\mathcal{S}^{k}}-{ {\bf x}}_{|\mathcal{S}^{k}}\right]&=\mathbb{E}\left[{\bf P}_{|\mathcal{S}^{k}}{\bf \Phi}_{|\mathcal{T}\setminus\mathcal{S}^{k}}{\bf x}_{|\mathcal{T}\setminus\mathcal{S}^{k}}\right] +\mathbb{E}\left[{\bf P}_{|\mathcal{S}^{k}}{\bf w}\right] \nonumber\\
&={\bf 0},
\end{align}
where the last equality follows from that all elements in ${\bf P}_{|\mathcal{S}^{k}}$, ${\bf \Phi}_{|\mathcal{T}\setminus\mathcal{S}^{k}}$, ${\bf x}_{|\mathcal{T}\setminus\mathcal{S}^{k}}$, and ${\bf w}$ are mutually independent and $\mathbb{E}[{\bf \Phi}_{|\mathcal{T}\setminus\mathcal{S}^{k}}]={\bf 0}$ and $\mathbb{E}[{\bf w}]={\bf 0}$. Next we compute the error covariance matrix. Conditioned that the sub-matrix ${\bf \Phi}_{|\mathcal{S}^{k}}$ is fixed, the error covariance matrix is  
\begin{align}
&\mathbb{E}\!\left[\!\left(\!{ {\bf x}}_{|\mathcal{S}^{k}}-{\hat {\bf x}}_{|\mathcal{S}^{k}}\!\right)\!\!\left(\!{ {\bf x}}_{|\mathcal{S}^{k}}-{\hat {\bf x}}_{|\mathcal{S}^{k}}\!\right)^T\!\mid {\bf \Phi}_{|\mathcal{S}^{k}}\right]\\ \nonumber 
&=\!{\bf P}_{\!|\mathcal{S}^{k}}{\bf \Phi}_{|\mathcal{T}\!\setminus\!\mathcal{S}^{k}}\mathbb{E}\!\left[\!{\bf x}_{|\mathcal{T}\setminus\mathcal{S}^{k}}{\bf x}_{|\mathcal{T}\setminus\!\mathcal{S}^{k}}^T\!\right]\!{\bf \Phi}_{\!|\mathcal{T}\!\setminus\!\mathcal{S}^{k}}^T{\bf P}_{|\mathcal{S}^{k}}^T \!+\! {\bf P}_{\!|\mathcal{S}^{k}}\mathbb{E}[\!{\bf w}\!{\bf w}^T]{\bf P}_{\!|\mathcal{S}^{k}}^T \\ \nonumber 
&\stackrel{(a)}{=}\!\sigma^2_x{\bf P}_{\!|\mathcal{S}^{k}}\mathbb{E}\left[{\bf \Phi}_{|\mathcal{T}\setminus\mathcal{S}^{k}} {\bf \Phi}_{\!|\mathcal{T}\setminus\mathcal{S}^{k}}^T\right] {\bf P}_{|\mathcal{S}^{k}}^T \!+\! \sigma^2_w{\bf P}_{\!|\mathcal{S}^{k}}{\bf P}_{\!|\mathcal{S}^{k}}^T \\ \nonumber 
&\stackrel{(b)}{=}\left(\frac{\sigma^2_x(K-k)}{M}+\sigma^2_w\right)\left({\bf \Phi}_{|\mathcal{S}^{k}}^T{\bf \Phi}_{|\mathcal{S}^{k}}\right)^{-1} 
\end{align}
where (a) is due to $\mathbb{E}[{\bf w}\!{\bf w}^T]=\sigma^2_w{\bf I}$ and $\mathbb{E}\!\left[\!{\bf x}_{|\mathcal{T}\setminus\mathcal{S}^{k}}{\bf x}_{|\mathcal{T}\setminus\!\mathcal{S}^{k}}^T\!\right]=\sigma^2_x{\bf I}$ and (b) follows from $\mathbb{E}\left[{\bf \Phi}_{|\mathcal{T}\setminus\mathcal{S}^{k}} {\bf \Phi}_{\!|\mathcal{T}\setminus\mathcal{S}^{k}}^T\right]=\frac{K-k}{M}{\bf I}$ and ${\bf P}_{\!|\mathcal{S}^{k}}{\bf P}_{\!|\mathcal{S}^{k}}^T\!=\!\left(\!{\bf \Phi}_{|\mathcal{S}^{k}}^T{\bf \Phi}_{|\mathcal{S}^{k}}\!\!\right)^{\!-1}$. 
Let ${\bf b}_i$ and ${\bf \Phi}_{|\mathcal{S}^{k}_i}$ be the $i$th column vector in ${\bf \Phi}_{|\mathcal{S}^{k}}$ and a submatrix obtained by eliminating ${\bf b}_i$ in ${\bf \Phi}_{|\mathcal{S}^{k}}$ where $i\in\mathcal{S}^{k}$. The $i$th diagonal element of $\left({\bf \Phi}_{|\mathcal{S}^{k}}^T{\bf \Phi}_{|\mathcal{S}^{k}}\!\!\right)^{\!-1}$ is given by
\begin{align}
\left[\left({\bf \Phi}_{|\mathcal{S}^{k}}^T{\bf \Phi}_{|\mathcal{S}^{k}}\!\right)^{\!-1}\right]_{i,i} =\frac{1}{{\bf b}_i^T{\bf P}^{\perp}_{{\bf \Phi}_{|\mathcal{S}^{k}_i}}{\bf b}_i}
\end{align} 
where ${\bf P}^{\perp}_{{\bf \Phi}_{|\mathcal{S}^{k}_i}}={\bf I}-{\bf \Phi}_{|\mathcal{S}^{k}_i}\!\!\left(\!{\bf \Phi}_{|\mathcal{S}^{k}_i}^T{\bf \Phi}_{|\mathcal{S}^{k}_i}\!\!\right)^{\!-1}\!\!{\bf \Phi}_{|\mathcal{S}^{k}_i}^T$ stands for the orthogonal projection onto the null space of ${\bf \Phi}_{|\mathcal{S}^{k}_i}$. Since all elements in ${\bf b}_i$ and ${\bf \Phi}_{|\mathcal{S}^{k}_i}$ are assumed to be IID Gaussian random variables $\mathcal{N}\left(0,\frac{1}{M}\right)$, $M{\bf b}_i^T{\bf P}^{\perp}_{{\bf \Phi}_{|\mathcal{S}^{k}_i}}{\bf b}_i$ is distributed as a Chi-squared random variable with degrees of freedom $M-k$, i.e., $M{\bf b}_i^T{\bf P}^{\perp}_{{\bf \Phi}_{|\mathcal{S}^{k}_i}}{\bf b}_i\sim \chi^2_{(M-k)}$. As a result, by marginalizing with respect to the Chi-squared distribution, we have the variance of the $i$th estimation error as
\begin{align}
\mathbb{E}[({\hat x}_i-{x}_i)^2]&=\left(\frac{\sigma^2_x(K-k)}{M}+\frac{\tilde\sigma^2_w}{M}\right)\mathbb{E}\left[\frac{1}{{\bf b}_i^T{\bf P}^{\perp}_{{\bf \Phi}_{|\mathcal{S}^{k}_i}}{\bf b}_i}
 \right] \nonumber \\
&=\left(\frac{\sigma^2_x(K-k)}{M}+\frac{\tilde\sigma^2_w}{M}\right) \frac{M}{M-k-2}\nonumber \\
&=\frac{\sigma^2_x(K-k)+\tilde\sigma^2_w}{M-k-2},
\end{align}
	which completes the proof.


\begin{thebibliography}{1}
% 



\bibitem{CandesTao2005}
E. J. Cand$\grave{\rm e}$s and T. Tao, \newblock ``Decoding by linear programming,''
\newblock {\em IEEE Trans. Inf. Theory,} vol. 51, no. 12, pp. 4203 - 4215, Dec. 2005.

\bibitem{CandesRombergTao2006}
E. J. Cand$\grave{\rm e}$s, J. Romberg, and T. Tao, \newblock ``Robust uncertainty principles: exact signal reconstruction from highly incomplete frequency information,''
\newblock {\em IEEE Trans. Inf. Theory,} vol. 52, no. 2, pp. 489-509, Feb. 2006.


\bibitem{CandesWakin_Mag}
E. J. Cand$\grave{\rm e}$s and M. B. Wakin, \newblock ``An introduction to compressive sampling,''
\newblock {\em IEEE Signal Proc. Mag.,} vol. 25, pp.
21-30, March 2008.


\bibitem{Eldar_book}
Y. C. Eldar and G. Kutyniok, \newblock ``Compressed sensing : theory and applications,''
\newblock {\em Cambridge Univ. Press,} 2012.

\bibitem{Garey_book}
M. R. Garey and D. S. Johnson, \newblock ``Computers and intractability: a guide to the theory of NP-completeness,''
\newblock {\em W. H. Freeman,} 1979.

 

\bibitem{CandesRomberg2007}
E. J. Cand$\grave{\rm e}$s and J. Romberg,
\newblock ``Sparsity and incoherence in compressive sampling,''
\newblock {\em Inverse problems,} vol. 23, pp. 969, Apr. 2007.


%
%\bibitem{Hassibi2009}
%M. Stojnic, F. Parvaresh, and B. Hassibi,
%\newblock ``On the reconstruction of
%block-sparse signals with an optimal number of measurements,''
%\newblock {\em IEEE
%Trans. Signal Processing,} vol. 57, no. 8, pp. 3075-3085, Aug. 2009.


\bibitem{Candes2008}
E. J. Cand$\grave{\rm e}$s,
\newblock ``The restricted isometry property and its implications for compressed sensing,''
\newblock {\em Comptes Rendus Mathematique,} vol. 346, no. 9-10, pp. 589-592, Feb. 2008.

%\bibitem{Eldar2009}
%Y. Eldar and M. Mishali, 
%\newblock ``Robust recovery of signals from a structured
%union of subspaces,''
%\newblock {\em IEEE
%Trans. Info. Theory,} vol. 55, no. 11, pp. 5302-5316, Nov. 2009.

\bibitem{Chen_BP}
S. S. Chen, \newblock ``Basis pursuit,''
\newblock {\em Ph.D. dissertation,} Stanford Univ., Stanford, CA, Nov. 1995.

\bibitem{Nesterov}
Y. Nesterov and A. Nemirovskii,\newblock ``Interior-point polynomial algorithms in convex programming.,''
\newblock {\em SIAM,} 1994. 


\bibitem{TroppGilbert2007}
J. A. Tropp and A. C. Gilbert, \newblock ``Signal recovery from random measurements via orthogonal matching pursuit,'' \newblock {\em IEEE Trans. Inf. Theory,}  vol. 53, no. 12, pp. 4655-4666, Dec. 2007.

   

\bibitem{DavenportWakin2010}
M. A. Davenport and M. B. Wakin, \newblock ``Analysis of orthogonal matching pursuit using the restricted isometry property,'' \newblock {\em IEEE Trans. Inf. Theory,} vol. 56, no. 9, pp. 4395-4401, Sept. 2010.

\bibitem{CaiWang2011}
T. T. Cai and L. Wang, \newblock ``Orthogonal matching pursuit for sparse signal recovery with noise,''
\newblock {\em IEEE Trans. Inf. Theory,} vol. 57, no. 7, pp. 4680-4688, July 2011.

\bibitem{Zhang}
T. Zhang, \newblock ``Sparse recovery with orthogonal matching pursuit under RIP,''
\newblock {\em IEEE Trans. Inf. Theory,} vol. 57, no. 9, pp. 6215-6221, Sept. 2011.

\bibitem{LiuTemlyakov2012}
E. Liu and V. N. Temlyakov, \newblock ``The orthogonal super greedy algorithm and applications in compressed sensing,'' \newblock {\em IEEE Trans. Inf. Theory,} vol. 58, no. 4, pp. 2040-2047, Apr. 2012.

\bibitem{Donoho}
D. L. Donoho, Y. Tsaig, I. Drori, and J. L. Starck, \newblock ``Sparse solution of underdetermined systems of linear equations by stagewise orthogonal matching pursuit,''
\newblock {\em IEEE Trans. Inf. Theory,} vol. 58, no. 2, pp. 1094-1121, Feb. 2012.


\bibitem{IHT2009}
T. Blumensath and M. E. Davies, \newblock ``Iterative hard thresholding for compressed sensing,''
\newblock {\em Applied and Computational Harmonic Analysis,} vol. 27, no. 3, pp. 265-274, Nov. 2009.



\bibitem{gOMP}
J. Wang, S. Kwon, and B. Shim, \newblock ``Generalized orthogonal matching pursuit,''
\newblock {\em IEEE Trans. Signal Process.,} vol. 60, no. 12, pp. 6202-6216, Dec. 2012.

\bibitem{Needell2010_CoSaMP}
D. Needell and J. A. Tropp, \newblock ``CoSaMP: iterative signal recovery from incomplete and inaccurate samples,''
\newblock {\em Commun. ACM,} vol. 53, no. 12,
pp. 93-100, Dec. 2010.

\bibitem{Dai2009_SP}
W. Dai and O. Milenkovic, \newblock ``Subspace pursuit for compressive sensing signal reconstruction,''
\newblock {\em IEEE Trans. Inf. Theory,} vol. 55, no. 5, pp. 2230-2249, May 2009.





\bibitem{Verdu}
S. Verdu, \newblock ``Multiuser detection,''
\newblock {\em Cambridge University Press,} 1998.


\bibitem{Tropp2004}
J. A. Tropp, \newblock ``Greed is good: Algorithmic results for sparse approximation,''
\newblock {\em IEEE Trans. Inf. Theory,} vol. 50, no. 10, pp. 2231-2242, Oct. 2004.

\bibitem{Wainwright2009}
M. J. Wainwright, \newblock ``Information-theoretic limits on sparsity recovery in the high-dimensional and noisy setting,''
\newblock {\em IEEE Trans. Inf. Theory,} vol. 55, no. 12, pp. 5728-5741, Dec. 2009.

\bibitem{Ji2008}
S. Ji, Y. Xue, and L. Carin, \newblock ``Bayesian compressive sensing,'' \newblock {\em IEEE Trans. Signal Process.,}  vol. 56, no. 6, pp. 2346-2356, Jun. 2008.


\bibitem{Schniter2009}
P. Schniter, L. C. Potter, and J. Ziniel, \newblock ``Fast Bayesian matching pursuit,''\newblock {\em in Proc. of IEEE Information Theory and Applications Workshop,}  pp. 326-333, Jan. 2008.



\bibitem{Zayyani2009}
H. Zayyani, M. Babaie-Zadeh, and C. Jutten, \newblock ``An iterative bayesian algorithm for sparse component analysis in presence of noise,''\newblock {\em IEEE Trans. on Signal Processing,} vol. 57, no. 11, pp. 4378-4390, Nov. 2009.


\bibitem{Herzet2010}
 C. Herzet and A. Dremeau, \newblock ``Bayesian pursuit algorithms,''\newblock {\em in Proc. IEEE European Signal Processing Conference (EUSIPCO),}  pp. 1474-1478, Aug. 2010.
 
\bibitem{Herzet2011}
A. Dremeau, C. Herzet, L. Daudet, \newblock ``Soft Bayesian pursuit algorithm for sparse representations,''\newblock {\em in Proc. IEEE Statistical Signal Processing Workshop (SSP),}  pp. 341-344, Jan. 2011.


\bibitem{Tian2009}
Z. Tian , G. Leus, and V. Lottici, \newblock ``Detection of sparse signals under finite-alphabet constraints,''\newblock {\em in Proc. IEEE Int. Conf. Acoust. Speech Signal Process (ICASSP),} pp.2349 -2352, Mar. 2009.

\bibitem{Vishwanath2013}
A. K. Das and S. Vishwanath, \newblock ``On finite alphabet compressive sensing,''\newblock {\em in Proc. IEEE Int. Conf. Acoust. Speech Signal Process (ICASSP),}   pp. 5890-5894, Mar. 2013. 

\bibitem{cvx}
M. Grant and S. Boyd, \newblock ``CVX: Matlab software for disciplined convex programming,''\newblock {\em version 2.1,}  http://cvxr.com/cvx, Mar. 2014.


\bibitem{Needell_Tropp_2008}
D. Needell, J. A. Tropp, and R. Vershynin,  \newblock ``Greedy signal recovery review,''\newblock {\em in Proc. IEEE 42nd Asilomar Conference on Signals, Systems, and Computers,} pp. 1048-1050, Oct. 2008.



\bibitem{Ward}
R. Ward, \newblock ``Compressed sensing with cross validation,''
\newblock {\em IEEE Trans. Inf. Theory,} vol. 55, no. 12, pp. 5773-5782, Dec. 2009.


\bibitem{Fletcher}
 A. K. Fletcher, S. Rangan, and V. K.Goyal, \newblock ``A sparsity detection framework for on-off random access channels,''
\newblock {\em IEEE ISIT 2009,} pp.169-173, 2009.
 

\bibitem{Wang_Yin}
Y. Wang and W. Yin, \newblock ``Sparse signal reconstruction via iterative support detection,''
\newblock {\em SIAM J. Imaging Sciences,} vol. 3, no. 3, pp. 462-491, 2010.

\end{thebibliography}
\end{document}